\documentclass[a4paper,12pt]{amsart}

\usepackage{a4wide,graphicx,amssymb,amscd,amsmath,latexsym,amsbsy,amsthm,color,multicol}
\usepackage{marginnote,mathtools,epsfig,enumitem,array,calc,rotating,bm,bbm,mathrsfs} 
\usepackage{subfig}
\usepackage[libertine,cmintegrals,cmbraces,vvarbb]{newtxmath}

\usepackage{etoolbox}

\usepackage{tikz}

\usetikzlibrary{decorations.markings}
\usetikzlibrary{calc} % coordinate calculation
\usetikzlibrary{arrows} % arrow heads
\usetikzlibrary{patterns} % hatch for wall
\usetikzlibrary{decorations.pathreplacing} % big braces

\usepackage[
   hmarginratio={1:1},     % equal left and right margins
   vmarginratio={1:1},     % equal top and bottom margins
   textwidth=457pt,        % new text width
   textheight=630pt,       % new text height
   heightrounded,          % always useful
]{geometry}

\makeatletter
\let\old@tocline\@tocline
\let\section@tocline\@tocline

\newcommand{\subsection@dotsep}{4.5}
\newcommand{\subsubsection@dotsep}{4.5}
\patchcmd{\@tocline}
  {\hfil}
  {\nobreak
     \leaders\hbox{$\m@th
        \mkern \subsection@dotsep mu\hbox{.}\mkern \subsection@dotsep mu$}\hfill
     \nobreak}{}{}
\let\subsection@tocline\@tocline
\let\@tocline\old@tocline

\patchcmd{\@tocline}
  {\hfil}
  {\nobreak
     \leaders\hbox{$\m@th
        \mkern \subsubsection@dotsep mu\hbox{.}\mkern \subsubsection@dotsep mu$}\hfill
     \nobreak}{}{}
\let\subsubsection@tocline\@tocline
\let\@tocline\old@tocline

\let\old@l@subsection\l@subsection
\let\old@l@subsubsection\l@subsubsection

\def\@tocwriteb#1#2#3{%
  \begingroup
    \@xp\def\csname #2@tocline\endcsname##1##2##3##4##5##6{%
      \ifnum##1>\c@tocdepth
      \else \sbox\z@{##5\let\indentlabel\@tochangmeasure##6}\fi}%
    \csname l@#2\endcsname{#1{\csname#2name\endcsname}{\@secnumber}{}}%
  \endgroup
  \addcontentsline{toc}{#2}%
    {\protect#1{\csname#2name\endcsname}{\@secnumber}{#3}}}%

\newlength{\@tocsectionindent}
\newlength{\@tocsubsectionindent}
\newlength{\@tocsubsubsectionindent}
\newlength{\@tocsectionnumwidth}
\newlength{\@tocsubsectionnumwidth}
\newlength{\@tocsubsubsectionnumwidth}
\newcommand{\settocsectionnumwidth}[1]{\setlength{\@tocsectionnumwidth}{#1}}
\newcommand{\settocsubsectionnumwidth}[1]{\setlength{\@tocsubsectionnumwidth}{#1}}
\newcommand{\settocsubsubsectionnumwidth}[1]{\setlength{\@tocsubsubsectionnumwidth}{#1}}
\newcommand{\settocsectionindent}[1]{\setlength{\@tocsectionindent}{#1}}
\newcommand{\settocsubsectionindent}[1]{\setlength{\@tocsubsectionindent}{#1}}
\newcommand{\settocsubsubsectionindent}[1]{\setlength{\@tocsubsubsectionindent}{#1}}

\renewcommand{\l@section}{\section@tocline{1}{\@tocsectionvskip}{\@tocsectionindent}{\@tocsectionnumwidth}{\@tocsectionformat}}%
\renewcommand{\l@subsection}{\subsection@tocline{1}{\@tocsubsectionvskip}{\@tocsubsectionindent}{\@tocsubsectionnumwidth}{\@tocsubsectionformat}}%
\renewcommand{\l@subsubsection}{\subsubsection@tocline{1}{\@tocsubsubsectionvskip}{\@tocsubsubsectionindent}{\@tocsubsubsectionnumwidth}{\@tocsubsubsectionformat}}%
\newcommand{\@tocsectionformat}{}
\newcommand{\@tocsubsectionformat}{}
\newcommand{\@tocsubsubsectionformat}{}
\expandafter\def\csname toc@1format\endcsname{\@tocsectionformat}
\expandafter\def\csname toc@2format\endcsname{\@tocsubsectionformat}
\expandafter\def\csname toc@3format\endcsname{\@tocsubsubsectionformat}
\newcommand{\settocsectionformat}[1]{\renewcommand{\@tocsectionformat}{#1}}
\newcommand{\settocsubsectionformat}[1]{\renewcommand{\@tocsubsectionformat}{#1}}
\newcommand{\settocsubsubsectionformat}[1]{\renewcommand{\@tocsubsubsectionformat}{#1}}
\newlength{\@tocsectionvskip}
\newcommand{\settocsectionvskip}[1]{\setlength{\@tocsectionvskip}{#1}}
\newlength{\@tocsubsectionvskip}
\newcommand{\settocsubsectionvskip}[1]{\setlength{\@tocsubsectionvskip}{#1}}
\newlength{\@tocsubsubsectionvskip}
\newcommand{\settocsubsubsectionvskip}[1]{\setlength{\@tocsubsubsectionvskip}{#1}}

\patchcmd{\tocsection}{\indentlabel}{\makebox[\@tocsectionnumwidth][l]}{}{}
\patchcmd{\tocsubsection}{\indentlabel}{\makebox[\@tocsubsectionnumwidth][l]}{}{}
\patchcmd{\tocsubsubsection}{\indentlabel}{\makebox[\@tocsubsubsectionnumwidth][l]}{}{}

\newcommand{\@sectypepnumformat}{}
\renewcommand{\contentsline}[1]{%
  \expandafter\let\expandafter\@sectypepnumformat\csname @toc#1pnumformat\endcsname%
  \csname l@#1\endcsname}
\newcommand{\@tocsectionpnumformat}{}
\newcommand{\@tocsubsectionpnumformat}{}
\newcommand{\@tocsubsubsectionpnumformat}{}
\newcommand{\setsectionpnumformat}[1]{\renewcommand{\@tocsectionpnumformat}{#1}}
\newcommand{\setsubsectionpnumformat}[1]{\renewcommand{\@tocsubsectionpnumformat}{#1}}
\newcommand{\setsubsubsectionpnumformat}[1]{\renewcommand{\@tocsubsubsectionpnumformat}{#1}}
\renewcommand{\@tocpagenum}[1]{%
  \hfill {\mdseries\@sectypepnumformat #1}}

\let\oldappendix\appendix
\renewcommand{\appendix}{%
  \leavevmode\oldappendix%
  \addtocontents{toc}{%
    \protect\settowidth{\protect\@tocsectionnumwidth}{\protect\@tocsectionformat\sectionname\space}%
    \protect\addtolength{\protect\@tocsectionnumwidth}{2em}}%
}
\makeatother

\makeatletter
\settocsectionnumwidth{2em}
\settocsubsectionnumwidth{2.5em}
\settocsubsubsectionnumwidth{3em}
\settocsectionindent{1pc}%
\settocsubsectionindent{\dimexpr\@tocsectionindent+\@tocsectionnumwidth}%
\settocsubsubsectionindent{\dimexpr\@tocsubsectionindent+\@tocsubsectionnumwidth}%
\makeatother

\settocsectionvskip{0pt}
\settocsubsectionvskip{-5pt}
\settocsubsubsectionvskip{-5pt}

\settocsectionformat{\bfseries}
\settocsubsectionformat{\mdseries}
\settocsubsubsectionformat{\mdseries}
\setsectionpnumformat{\bfseries}
\setsubsectionpnumformat{\mdseries}
\setsubsubsectionpnumformat{\mdseries}

\let\oldtableofcontents\tableofcontents
\renewcommand{\tableofcontents}{%
  \vspace*{-\linespacing}% Default gap to top of CONTENTS is \linespacing.
  \oldtableofcontents}

\numberwithin{equation}{section}
\theoremstyle{plain}
 
\newtheorem{thm}{Theorem}[section]

\newtheorem{prop}[thm]{Proposition}

\newtheorem{defi}[thm]{Definition}
\newtheorem{lem}[thm]{Lemma}
\newtheorem{cor}[thm]{Corollary}
\newtheorem{eg}[thm]{Example}

\theoremstyle{remark}
\newtheorem{rema}[thm]{Remark}
\newtheorem{conjec}[thm]{Conjecture}

\newcommand{\Z}{\mathbb{Z}} 

\newcommand{\R}{\mathbb{R}}
\newcommand{\C}{\mathbb{C}}
\newcommand{\V}{\mathbb{V}}

\newcommand{\AR}{\mathscr{A}(\mathcal{R})}

\newcommand{\dd}{\mathrm{d}}

\newcommand{\ii}{\mathrm{i}}
\newcommand{\ket}[1]{\left|#1\right\rangle}      % Ket-Zustand
\newcommand{\bra}[1]{\left\langle #1\right|}     % Bra-Zustand

\newcommand{\PROD}[3]{\mathop{\overrightarrow\prod}\limits_{#1 \le #2 \le #3 }}

\newcommand{\hypref}[2]{\ifx\href\asklfhas #2\else\href{#1}{#2}\fi}
\newcommand{\Secref}[1]{Section~\ref{#1}}

\newcommand{\Appref}[1]{Appendix~\ref{#1}}

\newcommand{\Figref}[1]{Figure~\ref{#1}}

\renewcommand{\eqref}[1]{(\ref{#1})}

\def\[{\begin{equation}}
\def\]{\end{equation}}
\def\<{\begin{eqnarray}}
\def\>{\end{eqnarray}}

\title[]{Six-vertex model and non-linear \\ differential equations I. Spectral problem}
\author{W. Galleas}

\address{Institut f\"ur Theoretische Physik, Eidgen\"ossische Technische Hochschule Z\"urich, Wolfgang-Pauli-Strasse 27, 8093 Z\"urich, Switzerland}
\email{galleasw@phys.ethz.ch}

\thanks{The work of W.G. is partially supported by the Swiss National Science Foundation through the NCCR SwissMAP}
\begin{document}

%%%%%%%%%%%%%%%%%%%%%%%%%%%%%%%
\begin{abstract}
In this work we relate the spectral problem of the toroidal six-vertex model's transfer matrix with the theory of integrable non-linear differential
equations. More precisely, we establish an analogy between the \emph{Classical Inverse Scattering Method} and previously proposed functional equations
originating from the Yang-Baxter algebra. The latter equations are then regarded as an \emph{Auxiliary Linear Problem} allowing us to show that the 
six-vertex model's spectrum solves Riccati-type non-linear differential equations. Generating functions of conserved quantities are expressed in terms of
determinants and we also discuss a relation between our Riccati equations and a stationary Schr\"odinger equation.

\end{abstract}
%%%%%%%%%%%%%%%%%%%%%%%%%%%%%%%

\maketitle

\tableofcontents

%%%%%%%%%%%%%%%%%%%%%%%%%%%%%%%%%%%%%%%%%%
\section{Introduction} \label{sec:INTRO}
%%%%%%%%%%%%%%%%%%%%%%%%%%%%%%%%%%%%%%%%%%%

Solving differential equations is one of the main problems one needs to deal with in order to study a large amount of physical theories. In particular, important physical phenomena
are in their turn described by non-linear differential equations. \emph{Solitons} are remarkable examples of such non-linear phenomena and their study have
precipitated a series of developments in both physics and mathematics. The \emph{Korteweg - de Vries} (KdV) equation \cite{KdV_1895} is a distinguished 
example of equation describing solitons and the problem of solving it exactly has led to the formulation of the \emph{Classical Inverse Scattering Method} (CISM)
by Gardner, Greene, Kruskal and Miura \cite{inverse_scattering_1967}. In fact, conditions allowing for the exact integration of differential 
equations originating in \emph{Hamiltonian systems} have been previously formulated by Liouville \cite{Liouville_1855}. As for 
generic \emph{differential systems} we have Frobenius' criteria stating that a system of differential equations, regarded as a collection
of \emph{differential forms} on a manifold $M$, is integrable if $M$ admits a foliation by maximal integral manifolds \cite{Frobenius_1877}. 

On the other hand, several quantum systems are not \emph{a priori} described by differential equations and, in such cases, one can not immediately
resort to Liouville's conditions. However, for quantum systems satisfying special properties we have available a quantum version of the CISM 
\cite{Sk_Faddeev_1979, Takh_Faddeev_1979} in which Bethe ansatz methods play a fundamental role.

\subsection{Bethe ansatz} 
The modern theory of integrable systems has been largely shaped by Bethe's seminal solution of the one-dimensional isotropic Heisenberg 
spin-chain with nearest-neighbors interaction \cite{Bethe_1931}. The latter is a paradigmatic model of quantum magnetism and it is also refereed to in the
literature as $\mathrm{XXX}$ model. Bethe's method contains a significant amount of physical insights and its starting point is the proposal of an
ansatz for the model's eigenvectors. Such ansatz is parameterized by additional variables which are then fine tuned by 
the eigenvalue problem of the associated Hamiltonian operator. 
Those variables are usually denominated \emph{Bethe roots}; and finding a consistent way of fixing such parameters
is one of the crucial steps in Bethe's method. The constraints imposed on Bethe roots are then called \emph{Bethe ansatz equations}.
This methodology pioneered by Bethe has been extended to a large number of spin-chain Hamiltonians and we shall not attempt to presenting a complete list
of models solved through it. However, it is important to remark the solutions of the $\mathrm{XXZ}$ 
\cite{Yang_Yang_1966, Yang_Yang_1966I, Yang_Yang_1966II, Yang_Yang_1966III} and Hubbard \cite{Lieb_Wu_1968} models as fundamental 
contributions to the understanding of the mathematical structure underlying Bethe ansatz method.

\subsection{Transfer matrix diagonalization} 
Although Bethe ansatz was originally devised for the diagonalization of spin-chain hamiltonians, it was later understood that its range of applications is
in fact much wider. For instance, in the series of works \cite{Lieb_1967, Lieb_1967a, Lieb_1967b, Lieb_1967c} Lieb has shown that Bethe ansatz
can also be employed in problems of classical Statistical Mechanics. More precisely, Lieb showed that the \emph{transfer matrix} associated with 
the six-vertex model on a torus can also be diagonalized by means of Bethe ansatz; and that such solution offers access to the model's free-energy in 
the thermodynamical limit.

For the sake of precision, we stress here that the six-vertex model solved by Lieb was not a generic one. The statistical weights assigned for 
vertex configurations in the latter were chosen in such a way that particular identities were fulfilled; ultimately granting notable 
properties to the model. Lieb's result played a fundamental role in subsequent developments in the field; for instance, in the development of Baxter's concept 
of \emph{commuting transfer matrices} and \emph{\textsc{t-q} equations} \cite{Baxter_1971}. The latter consists of functional relations characterizing the spectrum of commuting transfer 
matrices. 

As far as \textsc{t-q} equations method is concerned, some comments are in order. For instance, in addition to the model's transfer matrix,
\textsc{t-q} method introduces an auxiliary operator such that Bethe roots are identified with zeroes of its eigenvalues. This auxiliary operator
is usually referred to as \textsc{q}-operator and it is constructed ensuring commutativity with the transfer matrix one wants to diagonalize.
Additional properties are also required but we shall not go into such details.

Despite all the success of Bethe ansatz methods, it is fair to say there is one fundamental pitfall in its use. Although it yields a single 
compact formula encoding all eigenvalues of the model's transfer matrix or Hamiltonian; such formula still depends on Bethe roots which need
to be determined by solving Bethe ansatz equations. The latter consists of a set of algebraic equations whose space of solutions are not yet 
fully and rigorously understood.

\subsection{Functional methods} It is not a simple task to mensurate the importance of Baxter's \textsc{t-q} method. Besides the introduction of
functional and analytical methods in the theory of exactly solvable models of Statistical Mechanics, it has precipitated a series of developments
in several other fields. For instance, traces of certain monodromy matrices in conformal field theories have been identified with \textsc{q}-operators 
in \cite{Bazhanov_1997}. Moreover, \textsc{t-q} relations were shown to be central objects in the so called ODE/IM correspondence \cite{Dorey_2007}
and, more recently, Baxter \textsc{q}-operators have also made their appearance in quantum \textsc{k}-theory \cite{Pushkar_Smirnov_Zeitlin_2016}.

Nevertheless, it is still sensible to ask if \textsc{t-q} relations are the only functional equations characterizing the spectrum of commuting transfer
matrices. As a matter of fact, the so called \emph{fusion hierarchy} \cite{Kulish_Reshetikhin_Sklyanin_1981, Kirillov_Reshetikhin_1987} 
and the \emph{inversion relation} \cite{Stroganov_1979} provide alternative functional relations characterizing eigenvalues of 
transfer matrices within their own range of application. 

Functional relations originating from the fusion hierarchy are usually solved in terms of Bethe ansatz equations but it is worth 
remarking they consist of a set of relations involving extra transfer matrices in addition to the one initially intended for diagonalization.
On the other hand, inversion relations seem to be available only for free-fermion models. 
A novel type of functional relation was recently put forward in \cite{Galleas_Twists} describing the spectrum of the transfer matrix
associated with the trigonometric six-vertex model with anti-periodic boundary twists. In particular, when the six-vertex model anisotropy
parameter is a \emph{root-of-unity}, the equation presented in \cite{Galleas_Twists} truncates and it can be regarded as a generalization of 
Stroganov's inversion relations \cite{Stroganov_1979}.

\subsection{Algebraic-functional approach} The study of spectral problems associated with two-dimensional vertex models through functional 
equations has been a successful endeavor. This is in particular due to the variety of mechanisms allowing for the derivation of such 
equations. Although we have a few functional equations methods available,  \textsc{t-q} equations and their generalizations still play a distinguished
role among them since, at the end of the day, one is often led to \textsc{t-q} type relations. In addition to that, it is important
to remark that the implementation of such methods usually involve strong use of \emph{representation theoretical} properties
of the model under consideration.

An alternative functional method in the theory of exactly solvable models was put forward in \cite{Galleas_2008} for spectral problems and
in \cite{Galleas_2010} for partition functions with domain-wall boundaries. The main idea of \cite{Galleas_2008, Galleas_2010} is to use the
Yang-Baxter algebra, which is a common algebraic structure underlying integrable vertex models, as a source of functional relations characterizing 
quantities of interest. We refer to this approach as \emph{Algebraic-Functional (AF) method} and, by construction, it requires very little 
information on the particular representation we are considering. In this way, this method has resulted into very general types of functional
equations whose structure was shown to accommodate both partition functions with domain-wall boundary conditions and scalar products of Bethe vectors.

However, it is fair to say that the AF method, at its current stage, is not as well developed for spectral problems as it is for
partition functions with domain-wall boundaries. As for spectral problems we remark the following results of the AF method:
\begin{itemize}
\item relation between eigenvalues of the anti-periodic six-vertex model and the partition function of the six-vertex model with
domain-wall boundaries \cite{Galleas_Twists};
\\
\item derivation of partial differential equations underlying the spectrum of the six-vertex model with periodic boundary conditions \cite{Galleas_2015}.
\end{itemize}
As a matter of fact, the present work will be based on the functional equations originally derived in \cite{Galleas_2015}.

\subsection{This work} The range of applicability of the AF method has been extended over the past years and we refer 
the reader to the  works \cite{Galleas_2011, Galleas_2012, Galleas_2013, Galleas_SCP, Galleas_Lamers_2014, Galleas_Lamers_2015, Galleas_2016, Lamers_2015, Galleas_2016a, Galleas_2016b, Galleas_2016c}
for a detailed account. In particular, in the recent works \cite{Galleas_2016a, Galleas_2016b, Galleas_2016c} we have put forward a quite
general method for solving the kind of functional equations deduced from this approach. This new method overcomes several difficulties and, in 
particular, allows one to naturally express solutions as determinants.

The aforementioned works \cite{Galleas_2016a, Galleas_2016b, Galleas_2016c} focus on partition functions with domain-wall boundaries and 
scalar products of Bethe vectors; and here our goal is to extend that approach for studying the functional equation derived in
\cite{Galleas_2015}. Although the equations derived in \cite{Galleas_2015, Galleas_2016b, Galleas_2016c} share the same structure,
the one of \cite{Galleas_2015} still exhibits some fundamental differences. For instance, its coefficients encode eigenvalues of the six-vertex
model transfer matrix. This feature provides the initial insight for drawing an analogy between the AF method and the \emph{Classical Inverse Scattering (CIS) method}.
In particular, here we intend to show that the linear functional equation obtained in \cite{Galleas_2015} plays the same role
as the \emph{auxiliary linear problem} within the CIS framework.

The auxiliary linear problem in the CIS method can assume many shapes. For instance, 
\begin{itemize}
\item the problems $L(t) \psi = \lambda \psi$ and $\frac{\partial \psi}{\partial t} = M \psi$ in Lax representation
\[
\frac{\dd L}{\dd t} = [ M, L] \; ; \nonumber 
\]
\item the set of first-order equations $\frac{\partial F}{\partial x} = U(x, t) F$ and $\frac{\partial F}{\partial t} = V(x, t) F$ in the 
zero-curvature representation
\[
\frac{\partial U}{\partial t} -  \frac{\partial V}{\partial x} + [U , V] = 0 \; . \nonumber
\]
\end{itemize}
Non-linear differential equations one would like to solve are then encoded in such representations. 

In the present paper we put forward another type of auxiliary linear problem whose consistency condition encodes non-linear
functional equations describing quantities of interest.
A particular specialization of such functional equations then yields non-linear differential equations. Using the AF method we 
shall exhibit an explicit realization of such auxiliary linear problem encoding the spectrum of the six-vertex model. In particular, we
find that the spectrum of the six-vertex model transfer matrix is governed by Riccati non-linear differential equations and higher-order
analogues.

\subsection{Outline} 
The present paper is devoted to the study of the eigenvalue problem associated with the six-vertex model's transfer matrix. 
Although this problem has been extensively discussed in the literature, here we would like to offer a new perspective on 
the diagonalization of transfer matrices and establish a relation with the theory of non-linear differential equations.
In order to clarify the proposed relation between vertex models and non-linear differential equations, we have organized this paper
as follows. In \Secref{sec:TRANS} we describe some algebraic aspects of the six-vertex model which will be required throughout this work. 
\Secref{sec:TSPEC}, in its turn, is devoted to the description of the aforementioned spectral problem by means of functional equations
originating from the Yang-Baxter algebra. In particular, in \Secref{sec:TSPEC} we also show that such functional equations can be
regarded as an \emph{auxiliary linear problem} along the lines of \emph{Lax} and \emph{zero-curvature} representations. One of the most important
results of the present paper is then described in \Secref{sec:QUANT}. More precisely, in \Secref{sec:QUANT} we present a general procedure
for constructing conserved quantities underlying certain non-linear functional equations encoded in our version of auxiliary linear problem.
In \Secref{sec:AUX} we describe a method for solving the aforementioned auxiliary linear problem and implement it for the 
spectral problem associated with the six-vertex model. The method put forward in \Secref{sec:AUX} has several byproducts and one of them is showing that
eigenvalues of the six-vertex model's transfer matrix satisfies Riccati equations. \Secref{sec:PDE} is then devoted to another byproduct
of \Secref{sec:AUX}. More precisely, in \Secref{sec:PDE} we describe a set of non-linear \emph{Partial Differential Equations} (PDEs)
solved by quantities introduced in \Secref{sec:AUX}. The generating functions of conserved quantities described in \Secref{sec:QUANT}
are not restricted to the six-vertex model's eigenvalue problem and such specialization is then investigated in details in \Secref{sec:6VC}.
Next, in \Secref{sec:DISCR} we discuss the issue of discretization of the transfer matrix's spectrum. The combination
of the aforementioned results allows us to unveil a relation between the six-vertex model's spectral problem and a particular
stationary Schr\"odinger equation. The latter relation is then made precise in \Secref{sec:SCHROD} and concluding remarks are
discussed in \Secref{sec:CONCL}. Some extra results discussed in the main text are then gathered in \Appref{app:FUN}.

%%%%%%%%%%%%%%%%%%%%%%%%%%%%%%%%%%%%%%%%%%
\section{Six-vertex model's algebraic formulation} \label{sec:TRANS}
%%%%%%%%%%%%%%%%%%%%%%%%%%%%%%%%%%%%%%%%%%%

The origin of the six-vertex model is intimately related to the problem of the \emph{ice residual entropy} \cite{Pauling_1935}. 
The literature devoted to its study is quite extensive and we refer the reader to \cite{Baxter_book, Korepin_book} and references therein
for a more detailed account. Here we restrict our presentation to the algebraic aspects underlying the exactly solvable
six-vertex model which will be required throughout this work. 
In particular, in this section we shall consider a slight generalization of the results previously presented in \cite{Galleas_2015}. 
This generalization corresponds to a variation of strictly periodic boundary conditions also known as
\emph{boundary twists} \cite{deVega_1984}; which is intimately associated with automorphisms of the Yang-Baxter algebra.
In this work we shall also employ conventions already introduced in \cite{Galleas_2016c}.

\subsection{Yang-Baxter algebra} Let $\V_i$ with index $i \in \Z_{\geq 0}$ denote a complex vector space and let $\mathcal{L}_i \in \text{End}( \mathbb{V}_i )$ be a matrix with 
non-commutative entries. We then refer to the relation
\[ \label{yba}
\mathcal{R}_{i j} (x - y) \; \mathcal{L}_i (x)  \mathcal{L}_j (y) = \mathcal{L}_j (y) \mathcal{L}_i (x) \mathcal{R}_{i j} (x - y) \quad \in \text{End} (\V_i \otimes \V_j)
\]
as Yang-Baxter algebra and use $\AR$ to denote \eqref{yba} associated with a given operator $\mathcal{R}_{ij} \colon \C \to \text{End} (\V_i \otimes \V_j)$.
Representations of $\AR$ then consist of pairs $(\V_{\mathcal{Q}}, \mathcal{L})$ where $\V_{\mathcal{Q}}$ is a diagonalizable module and the entries of
$\mathcal{L}$ are meromorphic functions on $\C$ with values in $\text{End}( \V_{\mathcal{Q}})$.

\subsection{$\AR$-automorphisms and boundary twists} \label{sec:auto} 
Let $\Gamma_i \in \text{End} (\V_i)$ be an invertible matrix satisfying $\left[ \mathcal{R}_{i j} (x) , \; \Gamma_i \; \Gamma_j  \right] = 0$.
Thus one can readily show the map $\mathcal{L}_i (x) \mapsto \Gamma_i \; \mathcal{L}_i (x)$ is an automorphism of $\AR$. Matrices $\Gamma_i$ satisfying
the above described properties can be regarded as deviations of strictly periodic boundary conditions in exactly solvable vertex models \cite{deVega_1984}.
Elements $\Gamma_i$ are also refereed to in the literature as \emph{boundary twists}.

\subsection{Modules over $\AR$} \label{sec:modules}
Let $\V_{\mathcal{Q}}$ be a diagonalizable module and $\mathcal{L}_k \colon \C \to \text{End} (\V_k \otimes \V_{\mathcal{Q}})$ be meromorphic.
Representations of $\AR$ are given for instance by pairs $(\V_{\mathcal{Q}}, \mathcal{L}_0)$ such that $\mathcal{L}_k$ $(k = i,j)$ fulfills \eqref{yba} in 
$\text{End} (\V_i \otimes \V_j \otimes \V_{\mathcal{Q}})$. Pairs $(\V_{\mathcal{Q}}, \mathcal{L}_0)$ constitute $\AR$-modules
and one can readily show that $( \bigotimes_{i=1}^L \V_i , \widetilde{\mathcal{T}}_0)$ with 
\[ \label{mono}
\widetilde{\mathcal{T}}_0 (x) \coloneqq \PROD{1}{j}{L} \mathcal{R}_{0 j} (x - \mu_j) \qquad \quad L \in \Z_{>0} \; ,
\]
is an $\mathscr{A} (\mathcal{R})$-module. In the present paper we shall actually consider $\mathscr{A} (\mathcal{R})$
upon the automorphism described in \Secref{sec:auto}. In this way we shall use instead the $\mathscr{A} (\mathcal{R})$-module
$(\bigotimes_{i=1}^L \V_i, \mathcal{T}_0)$ with $\mathcal{T}_0 \coloneqq \Gamma_0 \widetilde{\mathcal{T}}_0$.

\subsection{Yang-Baxter equation} We further ask $\AR$ to be associative and this requires $\mathcal{R}$ to satisfy 
the relation
\< \label{ybe}
&& \mathcal{R}_{12} (x_1 - x_2) \mathcal{R}_{13} (x_1 - x_3) \mathcal{R}_{23} (x_2 - x_3) = \nonumber \\
&& \qquad \qquad \quad \mathcal{R}_{23} (x_2 - x_3) \mathcal{R}_{13} (x_1 - x_3) \mathcal{R}_{12} (x_1 - x_2)  
\>
in $\text{End} (\V_1 \otimes \V_2 \otimes \V_3)$. Relation \eqref{ybe} is the celebrated Yang-Baxter equation and a large 
literature is devoted to finding its solutions. See for instance \cite{Bazhanov_1984, Jimbo_1986b, Bazhanov_Shadrikov_1987, Ge_Wu_Xue_1991, Galleas_Martins_2006}.
As far as \eqref{yba} is concerned one can regard the entries of $\mathcal{R}$ as structure constants of the algebra $\AR$.
Within the context of classical vertex models of Statistical Mechanics one can also interpret $\mathcal{R}$ as a matrix encoding statistical weights
of allowed configurations of vertices \cite{Baxter_book}.

\subsection{The symmetric six-vertex model} \label{sec:6V}

In order to discuss the mechanism relating the spectrum of vertex models with non-linear differential equations we focus our analysis 
on the symmetric six-vertex model. The latter is a well studied exactly solvable model and this makes it a natural candidate for illustrating
this connection. The $\mathcal{R}$-matrix associated with the symmetric six-vertex model also intertwines tensor products of
evaluation modules of the $\mathcal{U}_q [\widehat{\mathfrak{gl}_2} ]$ quantum affine algebra. In that case we consider 
$\V_i = \V \cong \C^2$ and let
\[
e_1 \coloneqq \begin{pmatrix} 1 \\ 0 \end{pmatrix} \quad \mbox{and} \quad e_2 \coloneqq \begin{pmatrix} 0 \\ 1 \end{pmatrix}
\]
be standard basis vectors in $\C^2$. Next we write $\mathcal{R} \colon \C \to \text{End} (\V \otimes \V)$ explicitly as
\<
\label{rmat}
\mathcal{R} (x) = \begin{pmatrix} a(x) & 0 & 0 & 0 \\ 0 & b(x) & c(x) & 0 \\ 0 & c(x) & b(x) & 0 \\ 0 & 0 & 0 & a(x) \end{pmatrix} 
\>
with respect to the ordered basis $\{ e_1 \otimes e_1, e_1 \otimes e_2 , e_2 \otimes e_1 , e_2 \otimes e_2 \}$.
The non-null entries of \eqref{rmat} explicitly reads $a(x) \coloneqq \sinh{(x + \gamma)}$, $b(x) \coloneqq \sinh{(x)}$ and 
$c(x) \coloneqq \sinh{(\gamma)}$ with $x$ and $\gamma$ denoting complex parameters. The latter parameters are respectively refereed to as \emph{spectral} 
and \emph{anisotropy} parameters. 
Also, throughout this work we fix $\gamma$ and write $q \coloneqq e^{\gamma}$ in order to identify \eqref{rmat} as a
$\mathcal{U}_q [\widehat{\mathfrak{gl}_2} ]$-intertwiner.

As here $\V_i = \V \cong \C^2$  we can conveniently write 
\[ \label{ABCD}
\mathcal{L} (x) \eqqcolon \begin{pmatrix}  \mathcal{A}(x) & \mathcal{B}(x) \\ \mathcal{C}(x) & \mathcal{D}(x) \end{pmatrix} \; .
\]
Then we find a representation map $\pi \colon \AR \to \text{End}(\V_{\mathcal{Q}})$ through the identification of 
$\mathcal{L}$ with $\widetilde{\mathcal{T}}_0$ (or $\mathcal{T}_0$) previously defined in \Secref{sec:modules}. 

In order to consider $\mathcal{T}_0$ one first needs to investigate the automorphism described in \Secref{sec:auto} for the particular
$\mathcal{R}$-matrix \eqref{rmat}. This problem was studied in \cite{deVega_1984} and two distinct kinds of matrices matrices $\Gamma$ have
been found. One is diagonal while the second matrix is purely off-diagonal. 
Here we restrict our attention to the diagonal case $\Gamma = \text{diag}(\phi_1 , \phi_2)$ with $\phi_1 , \phi_2 \in \C^{\times}$.

\subsection{Highest-weight module} \label{sec:hwm}
We proceed with the construction of highest-weight modules and for that we first need to define \emph{singular vectors} in the
$\mathscr{A} (\mathcal{R})$-module $( \V_{\mathcal{Q}} , \mathcal{L} )$. 
Such vectors are then defined as non-zero elements $v_0 \in \V_{\mathcal{Q}}$ satisfying the condition $\mathcal{C} (x) v_0 = 0$ for all
$x \in \C$. The $\mathcal{R}$-matrix \eqref{rmat} has an underlying $\mathfrak{gl}_2$ algebra structure and here we shall
use $\mathfrak{h}$ to denote the corresponding Cartan subalgebra.
In this way we regard $\V_{\mathcal{Q}}$ as a diagonalizable $\mathfrak{h}$-module 
and say an element $v \in \V_{\mathcal{Q}}$ has $\mathfrak{h}$-weight $h$ if $H v = h v$ for all $H \in \mathfrak{h}$. 
As for $\mathfrak{gl}_2$, $\mathfrak{h}$ is one-dimensional, and we simply take $H \coloneqq E_{11} - E_{22}$ with matrix
units $E_{ij}$ defined through the action $E_{ij} (e_k) \coloneqq \delta_{j k} e_i$. 

Next we assign the weight $(h, \lambda_{\mathcal{A}}(x), \lambda_{\mathcal{D}}(x) )$ to an element $v \in \V_{\mathcal{Q}}$ if $v$ has $\mathfrak{h}$-weight 
$h$, $\mathcal{A}(x) v = \lambda_{\mathcal{A}}(x) v$ and $\mathcal{D}(x) v = \lambda_{\mathcal{D}}(x) v$. Given the above definitions, a 
highest-weight module is constituted of singular vectors $v_0 \in \V_{\mathcal{Q}}$ having weight $(h, \lambda_{\mathcal{A}}(x), \lambda_{\mathcal{D}}(x) )$.
Then, considering $\V_{\mathcal{Q}} = \V^{\otimes L}$, one can use \eqref{rmat}, \eqref{mono} and \eqref{ABCD} to show that $\ket{0} \coloneqq (e_1)^{\otimes L}$
is a highest-weight vector with $\mathfrak{h}$-weight $L$,
\[ \label{lambda}
\lambda_{\mathcal{A}} (x) \coloneqq \prod_{j=1}^L a(x - \mu_j) \qquad \text{and} \qquad \lambda_{\mathcal{D}} (x) \coloneqq \prod_{j=1}^L b(x - \mu_j)  \; .
\]

\subsection{Transfer matrix} \label{sec:TMAT}
A two-dimensional vertex model can be regarded as an \emph{edge-colored} graph $\mathcal{G}$ embedded in a two-dimensional 
lattice such that each vertex in $\mathcal{G}$ has degree four or one. In particular, let $g_i$ be a subgraph of $\mathcal{G}$ 
constituted of a degree four vertex, its four adjacent vertices (also having degree four each) and their four connecting edges.
Also, let $h_i$ denote a subgraph composed of two adjacent vertices (one having degree one) and their one connecting edge. In this way
we build our vertex model on a graph $\mathcal{G} = \mathcal{G}_{bulk} \cup \mathcal{G}_{boundary}$ such that $\mathcal{G}_{bulk} \supseteq g_i$ and
$\mathcal{G}_{boundary} \supseteq h_i$. 

Next we would like to associate a \emph{partition function} to the graph $\mathcal{G}$. For that we need to assign statistical weights to the
bulk subgraphs $g_i$. Statistical weights for $h_i$ are usually described by certain \emph{boundary vectors}. 
The vertex model's bulk partition function is then given by the product of all weights of subgraphs in $\mathcal{G}_{bulk}$ summed over all
possible \emph{edge-coloring}. If $\mathcal{G}_{boundary}$ is non-empty we also need to include contributions from the boundary
weights in order to having the model's partition function fully defined.

The choice of lattice embedding $\mathcal{G}$ also plays an important role when defining a vertex model. 
For instance, some choices even allows one to completely characterize the model's partition function in terms of linear functionals 
acting on given vectors spaces. Here we shall consider a cylindrical embedding for the symmetric six-vertex model. In that case one can write the
model's bulk partition function as a trace functional on the $\AR$-module $(\bigotimes_{i=1}^L \V_i, \mathcal{T}_0)$
defined in \Secref{sec:6V}.
More precisely, we shall embed $\mathcal{G}$ in the cylinder $\mathscr{C}_L \coloneqq \Z_{>0} \times \Z_{>0} / \{ (i,j+L) \sim (i,j)\}$
in such a way that the relevant algebraic object is the transfer matrix $\mathrm{T} \coloneqq (\mathrm{tr} \otimes \mathrm{Id}) \mathcal{T}_0$. 
As for our conventions, here we are using $\mathrm{Id}$ to denote the identity in $\bigotimes_{i=1}^L \V_i$. 
In this way, and keeping in mind \eqref{ABCD}, one can conveniently write
\[ \label{tmat}
\mathrm{T}(x) = \phi_1 \; \mathcal{A}(x) + \phi_2 \; \mathcal{D}(x) \qquad \in \text{End} ((\C^2)^{\otimes L}) \; .
\]
So far we have ignored the contributions of $\mathcal{G}_{boundary}$ and that is justifiable if we further fold $\mathscr{C}_L$ into
the torus $\mathscr{T}_{M L} \coloneqq \mathscr{C}_L / \{ (i+M,j) \sim (i,j) \}$. In that case our partition function reads
$Z = \mathrm{tr}^{\otimes L} \; ( \mathrm{T}(x_1) \mathrm{T}(x_2) \dots \mathrm{T}(x_M))$ and the problem of computing $Z$ can be formulated as the 
eigenvalue problem of the transfer matrix \eqref{tmat}.

%%%%%%%%%%%%%%%%%%%%%%%%%%%%%%%%%%%%%%%%%%
\section{Spectral problem and functional equations} \label{sec:TSPEC}
%%%%%%%%%%%%%%%%%%%%%%%%%%%%%%%%%%%%%%%%%%%

In \Secref{sec:TMAT} we have described the partition function of the toroidal six-vertex model through the action of a linear functional on 
products of transfer matrices. In particular, our linear functional takes the form of a trace and such formulation maps the problem of evaluating
the model's partition function to an eigenvalue problem. It is worth remarking that this approach goes back to Kramers and Wannier works on the 
two-dimensional Ising model \cite{Kramers_1941a, Kramers_1941b}. 
The eigenvalue problem of the transfer matrix \eqref{tmat} has been tackled through Bethe ansatz in several formulations, see for instance 
\cite{Lieb_1967, Takh_Faddeev_1979, Reshet_1987, Sklyanin_1985b}, and here we propose an alternative way of dealing with that same problem.
Our method can be regarded as an extension of the AF approach previously put forward in \cite{Galleas_2015} for the symmetric 
six-vertex model and, in what follows, we shall describe the derivation of linear functional equations encoding the eigenvalue problem of the 
transfer matrix \eqref{tmat}.

\subsection{Auxiliary linear problem} \label{sec:ALP}

The mechanism we shall describe here has a direct counterpart in the Classical Inverse Scattering Method (CISM). 
In fact, this analogy will assist us through our analysis of the transfer matrix spectral problem. In order to make clearer statements,
let us first elaborate on the procedure usually employed within the CISM. The latter aims to produce exact solutions of evolution problems described by (non-linear) differential equations which can
be expressed as compatibility conditions between certain linear problems.
For instance, in the original proposal of the CISM by Gardner, Greene, Kruskal and Miura \cite{inverse_scattering_1967},
the authors have shown that the KdV equation can be reformulated as the compatibility condition between
two linear differential equations -- one of them being the Schr\"odinger equation. 
More precisely, such embedding uses solutions of the KdV equation as the potential function entering the linear Schr\"odinger equation.
In more general settings this type of auxiliary linear problem gives rise to the so called \emph{Lax pair} \cite{Lax_1968}.

In the present paper we shall investigate the use of the AF method as a source of auxiliary linear problems. A schematic description of the
announced analogy between the CISM and the AF approach can be found in \Figref{fig:analogy}. 
Interestingly, at the end of the day one naturally finds within our approach non-linear differential equations describing eigenvalues of 
the transfer matrix \eqref{tmat}. Such non-linear equations also emerge as the compatibility condition of our proposed auxiliary linear problem.

\begin{prop}[Auxiliary Linear Problem] \label{aux}
Let $\mathfrak{S}_{n+1}$ denote the symmetric group of degree $n+1$ on $\{ x_0 , x_1 , \dots , x_n  \}$ and let 
$\pi_{i,j} \in \mathfrak{S}_{n+1}$ be a $2$-cycle acting as permutation of variables $x_i$ and $x_j$. In addition to that, let
$\mathfrak{S}_{n} \subset \mathfrak{S}_{n+1}$ act on  $\{ x_1 , x_2 , \dots , x_n  \}$ and write 
$\mathcal{F}_n \in \C \llbracket x_1^{\pm 1} , x_2^{\pm 1} , \dots , x_n^{\pm 1} \rrbracket^{\mathfrak{S}_n}$
for a symmetric function on $\C^n$ satisfying the linear equation
\[ \label{MF}
\sum_{i=0}^n \mathrm{M}_i \; \mathcal{F}_n (x_0, x_1 , \dots , \widehat{x_i} , \dots , x_n) = 0
\]
for given coefficients $\mathrm{M}_i = \mathrm{M}_i (x_0, x_1 , \dots , x_n) \in \C \llbracket x_0^{\pm 1} , x_1^{\pm 1} , \dots , x_n^{\pm 1} \rrbracket$.
Then \eqref{MF} naturally extends to the system of equations
\[ \label{MFF}
\sum_{i=0}^n \mathrm{M}_{i,j} \; \mathcal{F}_n (x_0, x_1 , \dots , \widehat{x_i} , \dots , x_n) = 0 \qquad \qquad j \in \{ 0,1, \dots , n \}
\]
with coefficients
\< \label{mij}
\mathrm{M}_{i,j} = \begin{cases}
\pi_{0, j} \mathrm{M}_j \,\qquad i = 0 \cr
\pi_{0, j} \mathrm{M}_0 \qquad i = j \cr
\pi_{0, j} \mathrm{M}_i \qquad \text{otherwise} 
\end{cases} \; .
\>
\end{prop}
\begin{proof}
Straightforward application of $\pi_{0,j}$ on \eqref{MF} taking into account that $\mathcal{F}_n$ is a symmetric function.
\end{proof}
\begin{rema}
For particular coefficients $\mathrm{M}_i$, there is still the possibility that not all equations in \eqref{MFF} are linearly independent . 
\end{rema}

\begin{lem}[Compatibility condition] \label{det}
The system of linear  equations \eqref{MFF} is compatible iff
\[ \label{CC}
\mathrm{det} \left( \mathrm{M}_{i,j}  \right)_{0 \leq i, j \leq n} = 0 \; .
\]
\end{lem}
\begin{proof}
Equations \eqref{MFF} can be written in matricial form as $\mathcal{M} \; \vec{\mathcal{F}} = 0$, where $\mathcal{M}$ is a matrix of 
dimension $(n+1) \times (n+1)$ with entries $\mathrm{M}_{i,j}$ and $\vec{\mathcal{F}}$ is a column-vector having $\mathcal{F}_n (x_0, x_1 , \dots , \widehat{x_i} , \dots , x_n)$
at its $i$-th position. Thus \eqref{MFF} has non-trivial solution iff $\mathrm{det} (\mathcal{M}) = 0$.
\end{proof}

\subsubsection{Non-linear functional equations} \label{nlfe}
Suppose explicit formulae for the coefficients $\mathrm{M}_{i}$ are given in terms
of \emph{dependent variables} $\Lambda (x_i)$. In this way the compatibility condition \eqref{CC} gives rise to a non-linear functional
equation for the function $\Lambda$. Hence, the functional equation \eqref{MF} can be regarded as a linear embedding of the non-linear problem
describing such function. This mechanism is analogous to the one employed within the CISM and in the present work we shall present explicit
realizations of \eqref{MF}. In particular, our coefficients $\mathrm{M}_{i}$ will encode eigenvalues of the transfer matrix \eqref{tmat}.

\begin{figure} \centering
\scalebox{1}{
\begin{tikzpicture}[>=stealth]
\path (-0.4,-1) node[rectangle,fill=gray!20!white,draw,align=center] (p1) {\textsc{Classical Inverse Scattering} \\ \textsc{Method}}
      (-0.4,-2.5) node[rectangle,rounded corners=7pt,fill=gray!20!white,draw,align=center] (p2) {Auxiliary Linear \\ Problems}
      (-0.4,-4.5) node[rectangle,rounded corners=7pt,fill=gray!20!white,draw,align=center] (p3) {Lax/Zero-curvature \\ equation}
      (-0.4,-6.5) node[rectangle,rounded corners=7pt,fill=gray!20!white,draw,align=center] (p4) {Non-Linear Differential \\ Equation};
\begin{scope}[xshift = 7.2cm]
\path (0,-1) node[rectangle,fill=gray!20!white,draw,align=center] (p5) {\textsc{Algebraic-Functional} \\ \textsc{Method}}
      (0,-2.5) node[rectangle,rounded corners=7pt,fill=gray!20!white,draw,align=center] (p6) {Proposition \ref{aux}}
      (0,-4.5) node[rectangle,rounded corners=7pt,fill=gray!20!white,draw,align=center] (p7) {Determinantal eq. \\ Lemma \ref{det}}
      (0,-6.5) node[rectangle,rounded corners=7pt,fill=gray!20!white,draw,align=center] (p8) {Non-Linear Functional \\ Equation};
\end{scope}
% \draw [thick] (p1.south) -- (p2.north);
\draw [->, thick]  (p2.south) -- (p3.north);
\draw [->, thick]  (p3.south) -- (p4.north);

% \draw [thick] (p5.south) -- (p6.north);
\draw [->, thick]  (p6.south) -- (p7.north);
\draw [->, thick]  (p7.south) -- (p8.north);

\draw [<->, shorten <=0.2cm, shorten >=0.2cm, thick] (p2.east) -- (p6.west);
\draw [<->, shorten <=0.2cm, shorten >=0.2cm, thick] (p3.east) -- (p7.west);
\draw [<->, shorten <=0.2cm, shorten >=0.2cm, thick] (p4.east) -- (p8.west);
\end{tikzpicture}}
\caption{Analogy between the classical inverse scattering method and the algebraic-functional framework.}
\label{fig:analogy}
\end{figure}
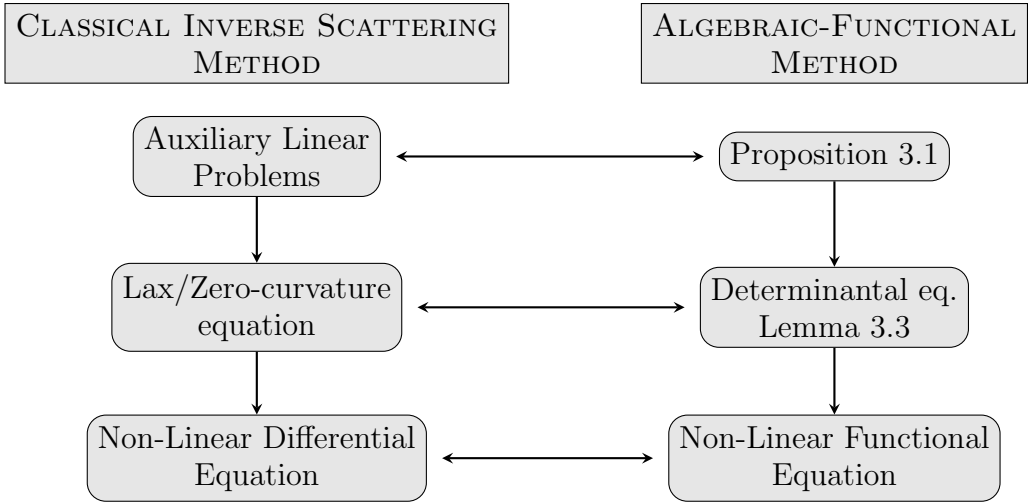

\subsection{Algebraic-Functional method} \label{sec:AFM}
The algebra $\mathscr{A} (\mathcal{R})$ is one of the corner stones of the algebraic Bethe ansatz method and the AF framework 
can be viewed as an alternative way of exploiting it. As for the symmetric six-vertex model, it is an algebra over $\C$ generated by elements $\mathcal{A}$, $\mathcal{B}$, $\mathcal{C}$ and 
$\mathcal{D}$ defined in \eqref{ABCD}. Moreover, here we shall regard \eqref{yba} as a matrix algebra with elements in 
$\C \llbracket x, x^{-1} \rrbracket \otimes \text{End} (\V^{\otimes L})$ which we also refer to as $\mathscr{A}_2 (\mathcal{R})$.
In this way, we write $\mathscr{M}_n \coloneqq \{ \mathcal{A} , \mathcal{B} , \mathcal{C} , \mathcal{D} \} (x_n)$ such that the repeated use of
$\mathscr{A} (\mathcal{R})$ produces relations in $\mathscr{A}_n (\mathcal{R}) \cong \mathscr{A}_{n-1} (\mathcal{R}) \otimes \mathscr{M}_n / \mathscr{A}_2 (\mathcal{R})$
for $n>2$. Next we look for a suitable linear functional $\Phi \colon \mathscr{A}_n (\mathcal{R}) \to \C [x_1^{\pm 1} , x_2^{\pm 1} , \dots , x_n^{\pm 1}] $
allowing us to use $\mathscr{A}_n (\mathcal{R})$ as a source of functional equations. Here we are interested in describing the spectrum of 
the transfer matrix \eqref{tmat}, and the direct inspection of relations in $\mathscr{A}_n (\mathcal{R})$ suggests looking for a linear functional
$\Phi$ satisfying the property
\[ \label{MAP}
\Phi (J_n) = \omega_J (x_1, x_2, \dots , x_n) \; \Phi (J_{n-1}) 
\]
for fixed meromorphic functions $\omega_J$ associated with particular elements $J_n \subseteq \mathscr{A}_n (\mathcal{R})$.

\subsubsection{Subalgebra $\mathscr{S}_{\mathcal{A} , \mathcal{B}}$}

In order to apply the AF method along the lines discussed in \Secref{sec:ALP} and \Secref{sec:AFM}, it is convenient to first delimit
subalgebras of $\AR$ which will be relevant for studying the spectrum of the transfer matrix \eqref{tmat}. Although $\AR$ contains 
four generators only two of them, namely $\mathcal{A}$ and $\mathcal{D}$, appears in the construction of the transfer matrix \eqref{tmat}.
Hence, we first consider the subalgebra $\mathscr{S}_{\mathcal{A} , \mathcal{B}} \subset \mathscr{A}_2 (\mathcal{R})$ formed by the following relations:
\< \label{AB}
\mathrm{F} (x_1) \mathrm{F} (x_2) &=& \mathrm{F} (x_2) \mathrm{F} (x_1) \qquad \mathrm{F} \in \{ \mathcal{A} , \mathcal{B}  \} \nonumber \\
\mathcal{A} (x_1) \mathcal{B} (x_2) &=& \frac{a(x_2 - x_1)}{b(x_2 - x_1)} \mathcal{B} (x_2) \mathcal{A} (x_1) + \frac{c(x_1 - x_2)}{b(x_1 - x_2)} \mathcal{B} (x_1) \mathcal{A} (x_2) \nonumber \\
\mathcal{B} (x_1) \mathcal{A} (x_2) &=& \frac{a(x_2 - x_1)}{b(x_2 - x_1)} \mathcal{A} (x_2) \mathcal{B} (x_1) + \frac{c(x_1 - x_2)}{b(x_1 - x_2)} \mathcal{A} (x_1) \mathcal{B} (x_2) \; .
\>
\begin{rema}
$\mathscr{S}_{\mathcal{A} , \mathcal{B}}$ comprises four relations out of the sixteen contained in $\mathscr{A}_2 (\mathcal{R})$.
\end{rema}

For latter convenience we also write $X \coloneqq \{ x_1 , x_2 , \dots , x_n \}$. Here $n \in \Z_{\geq 0}$ and we fix $L \in \Z_{\geq 1}$ such
that $n \leq L$. Also, let us introduce the shorthand notation
\[
X_{\alpha_1 , \alpha_2 , \dots , \alpha_l}^{\beta_1 , \beta_2 , \dots , \beta_m} \coloneqq X \cup \{ x_{\beta_1}, x_{\beta_2} , \dots , x_{\beta_m}  \} \backslash \{ x_{\alpha_1}, x_{\alpha_2} , \dots , x_{\alpha_l}  \}    \qquad l, m \in \Z_{\geq 0} \; .
\]
In this way we find the relation
\< \label{ABn}
\mathcal{A} (x_0) \; \mathcal{Y}(X) &=& \prod_{j=1}^n \frac{a(x_j - x_0)}{b(x_j - x_0)} \mathcal{Y}(X) \; \mathcal{A} (x_0) \nonumber \\
&& \quad + \; \sum_{i=1}^n \frac{c(x_0 - x_i)}{b(x_0 - x_i)} \prod_{\substack{j=1 \\ j \neq i}}^n \frac{a(x_j - x_i)}{b(x_j - x_i)} \mathcal{Y}(X_i^0) \; \mathcal{A} (x_i)
\>
in $\mathscr{A}_{n+1} (\mathcal{R})$ through the iteration of $\mathscr{S}_{\mathcal{A} , \mathcal{B}}$. 
The \emph{word} $\mathcal{Y}(X)$ in \eqref{ABn} is defined as
\[ \label{YX}
\mathcal{Y}(X) \coloneqq \prod_{x \in X} \mathcal{B} (x) \; .
\]
\begin{rema}
Although $\mathscr{A}_n (\mathcal{R})$ is non-abelian, the subalgebra containing only $\mathcal{B}$'s is abelian. Hence, we do not need to worry
about ordering in \eqref{YX}.
\end{rema}

\subsubsection{Subalgebra $\mathscr{S}_{\mathcal{D} , \mathcal{B}}$}

We then proceed by including the operator $\mathcal{D}$ in our analysis as it is also a constituent of the transfer matrix \eqref{tmat}. 
For that we additionally single out the subalgebra $\mathscr{S}_{\mathcal{D} , \mathcal{B}} \subset \mathscr{A}_2 (\mathcal{R})$
generated by the following relations:
\< \label{DB}
\bar{\mathrm{F}} (x_1) \bar{\mathrm{F}} (x_2) &=& \bar{\mathrm{F}} (x_2) \bar{\mathrm{F}} (x_1) \qquad \bar{\mathrm{F}} \in \{ \mathcal{D} , \mathcal{B}  \} \nonumber \\
\mathcal{D} (x_1) \mathcal{B} (x_2) &=& \frac{a(x_1 - x_2)}{b(x_1 - x_2)} \mathcal{B} (x_2) \mathcal{D} (x_1) + \frac{c(x_2 - x_1)}{b(x_2 - x_1)} \mathcal{B} (x_1) \mathcal{D} (x_2) \nonumber \\
\mathcal{B} (x_1) \mathcal{D} (x_2) &=& \frac{a(x_1 - x_2)}{b(x_1 - x_2)} \mathcal{D} (x_2) \mathcal{B} (x_1) + \frac{c(x_2 - x_1)}{b(x_2 - x_1)} \mathcal{D} (x_1) \mathcal{B} (x_2) \; .
\>

\begin{rema}
The subalgebra $\mathscr{S}_{\mathcal{D} , \mathcal{B}}$ also encloses four relations.
\end{rema}

The repeated iteration of $\mathscr{S}_{\mathcal{D} , \mathcal{B}}$ leaves us with the following relation in $\mathscr{A}_{n+1} (\mathcal{R})$,
\< \label{DBn}
\mathcal{D} (x_0) \; \mathcal{Y}(X) &=& \prod_{j=1}^n \frac{a(x_0 - x_j)}{b(x_0 - x_j)} \mathcal{Y}(X) \; \mathcal{D} (x_0) \nonumber \\
&& \quad + \; \sum_{i=1}^n \frac{c(x_i - x_0)}{b(x_i - x_0)} \prod_{\substack{j=1 \\ j \neq i}}^n \frac{a(x_i - x_j)}{b(x_i - x_j)} \mathcal{Y}(X_i^0) \; \mathcal{D} (x_i) \; .
\>

\subsubsection{The functional $\Phi$}

As far as relations \eqref{ABn} and \eqref{DBn} are concerned, there is not much difference in what we have done so far and the
methodology used for the Algebraic Bethe Ansatz. The difference between both approaches starts with the introduction of the functional
$\Phi \colon \mathscr{A}_{n+1} (\mathcal{R}) \to \C [x_0^{\pm 1} , x_1^{\pm 1} , \dots , x_n^{\pm 1}] $ in order to study the eigenvalue
problem of the operator \eqref{tmat}.  
For that we first notice that the transfer matrix $\mathrm{T} (x_0)$, as defined in \eqref{tmat}, lives in $\text{span}(\mathscr{M}_0)$. 
Next we remark that, although $\mathscr{A}_2 (\mathcal{R})$ has no subalgebra generated solely by $\{ \mathrm{T} , \mathcal{B} \}$,
one can still combine \eqref{ABn} and \eqref{DBn} in order to find appropriate relations in $\mathscr{A}_{n+1} (\mathcal{R})$.
Therefore, we multiply \eqref{ABn} by $\phi_1$ and add it to \eqref{DBn} multiplied by $\phi_2$ to obtain
\< \label{TBn}
&& \mathrm{T} (x_0) \; \mathcal{Y}(X) = \nonumber \\
&& \mathcal{Y}(X) \left[ \phi_1 \prod_{j=1}^n \frac{a(x_j - x_0)}{b(x_j - x_0)}  \; \mathcal{A} (x_0) + \phi_2 \prod_{j=1}^n \frac{a(x_0 - x_j)}{b(x_0 - x_j)} \; \mathcal{D} (x_0) \right] \nonumber \\
&&  + \; \sum_{i=1}^n \frac{c(x_0 - x_i)}{b(x_0 - x_i)} \mathcal{Y}(X_i^0) \left[ \phi_1 \prod_{\substack{j=1 \\ j \neq i}}^n \frac{a(x_j - x_i)}{b(x_j - x_i)}  \; \mathcal{A} (x_i) - \phi_2 \prod_{\substack{j=1 \\ j \neq i}}^n \frac{a(x_i - x_j)}{b(x_i - x_j)} \; \mathcal{D} (x_i) \right] \; . \nonumber \\
\>

\begin{rema}
Relation \eqref{TBn} belongs to $\mathscr{A}_{n+1} (\mathcal{R})$ and its action on the vector $\ket{0}$, as defined in \Secref{sec:hwm}, yields the so
called \emph{off-shell Bethe ansatz relation} within the context of the Algebraic Bethe Ansatz.
\end{rema}

\begin{prop}[Realization of $\Phi$]
Let $\ket{\Psi}$ be an eigenvector of $\mathrm{T}(x_0)$ with eigenvalue $\Lambda (x_0)$. More precisely, assume 
$\ket{\Psi}$ solves the equation $\mathrm{T}(x_0) \ket{\Psi} = \Lambda(x_0) \ket{\Psi}$ and write $\bra{\Psi}$ for its dual vector. Then the map 
\[ \label{PJ}
\Phi (J_{n+1}) \coloneqq \bra{\Psi} J_{n+1} \ket{0}
\]
with $J_{n+1} \subseteq \mathscr{A}_{n+1} (\mathcal{R})$ identified with \eqref{TBn} satisfies property \eqref{MAP}. 
\end{prop}
\begin{proof}
Recall $\ket{0}$ is a highest-weight vector as described in \Secref{sec:hwm} and use that $\bra{\Psi}$ is a dual eigenvector of $\mathrm{T}(x_0)$.
\end{proof}

\begin{rema}
Due to \eqref{yba} the six-vertex model transfer matrix satisfies $[ \mathrm{T}(x) , \mathrm{T}(y) ] = 0$ for all $x, y \in \C$. Consequently we can 
assume $\ket{\Psi}$ is $x_0$-independent.
\end{rema}

\subsubsection{Linear Functional Relation} 
As anticipated in \Secref{nlfe} here we intend to exhibit a realization of Proposition \ref{aux}. In particular, we shall describe coefficients 
$\mathrm{M}_i$ encoding transfer matrix's eigenvalues $\Lambda$. For that, in addition to the coefficients $\mathrm{M}_i$, we also need to ensure the existence
of a symmetric function $\mathcal{F}_n$ satisfying \eqref{MF}.

\begin{thm} \label{6V}
There exists $\mathcal{F}_n \in \C \llbracket x_1^{\pm 1} , x_2^{\pm 1} , \dots , x_n^{\pm 1} \rrbracket^{\mathfrak{S}_n}$ satisfying \eqref{MF}
with coefficients 
\< \label{MI}
&& \mathrm{M}_i \coloneqq \nonumber \\
&& \begin{cases}
\displaystyle \phi_1 \prod_{j=1}^n \frac{a(x_j - x_0)}{b(x_j - x_0)}  \; \lambda_{\mathcal{A}} (x_0) + \phi_2 \prod_{j=1}^n \frac{a(x_0 - x_j)}{b(x_0 - x_j)} \; \lambda_{\mathcal{D}} (x_0) - \Lambda(x_0) \qquad \qquad \; i = 0 \cr
\displaystyle \frac{c(x_0 - x_i)}{b(x_0 - x_i)} \left[ \phi_1 \prod_{\substack{j=1 \\ j \neq i}}^n \frac{a(x_j - x_i)}{b(x_j - x_i)}  \; \lambda_{\mathcal{A}} (x_i) - \phi_2 \prod_{\substack{j=1 \\ j \neq i}}^n \frac{a(x_i - x_j)}{b(x_i - x_j)} \; \lambda_{\mathcal{D}} (x_i) \right] \qquad \text{otherwise}
\end{cases} \nonumber \\
\>
and it is given by $\mathcal{F}_n (x_1 , x_2 , \dots , x_n) = \bra{\Psi} \mathcal{Y} (X) \ket{0}$.
\end{thm}

\begin{proof}
Since $[ \mathrm{T}(x) , \mathrm{T}(y) ] = 0$ we can take $\bra{\Psi}$ independent of spectral parameters. Also, we notice
$\pi_{i,j} \mathcal{Y} (X) = \mathcal{Y} (X)$ which implies that $\mathcal{F}_n (x_1 , x_2 , \dots , x_n) = \bra{\Psi} \mathcal{Y} (X) \ket{0}$
lives in $\C \llbracket x_1^{\pm 1} , x_2^{\pm 1} , \dots , x_n^{\pm 1} \rrbracket^{\mathfrak{S}_n}$. Next we apply the functional 
$\Phi$ defined in \eqref{PJ} to the Yang-Baxter algebra relation \eqref{TBn} taking into account that $\ket{0}$ is a highest-weight vector with weight
$(L , \lambda_{\mathcal{A}}, \lambda_{\mathcal{D}})$. The functions $\lambda_{\mathcal{A} , \mathcal{D}}$ have been defined in 
\eqref{lambda}. By doing so we find relation \eqref{MF} with coefficients $\mathrm{M}_i$ defined by \eqref{MI}. This completes our proof.
\end{proof}

\subsubsection{Non-linear functional equation for $\Lambda$}
As described in Proposition \ref{aux} one can extend \eqref{MI} to coefficients of a system of linear equations through
the action of permutations $\pi_{0,j}$. The extended coefficients are then found from \eqref{mij} and the compatibility condition
stated in Lemma \ref{det} yields a non-linear functional equation governing the eigenvalue $\Lambda$. 

\begin{rema} \label{rema:lab}
It is important to remark that we are not being rigorous with our notation. 
The transfer matrix $\mathrm{T}(x_0)$ commutes with the $\mathfrak{gl}_2$ Cartan subalgebra $\mathfrak{h}$ and this implies 
$\V_{\mathcal{Q}}$ decomposes into $\mathfrak{h}$-modules. Hence, we should have added an appropriate label to the 
transfer matrix's eigenvector $\ket{\Psi}$ in order to identify the particular $\mathfrak{h}$-module where $\ket{\Psi}$ lives. Consequently, 
the same label should also have been used for the corresponding eigenvalue $\Lambda$ but we will refrain to do so for the sake of simplicity.
\end{rema}

The linear problem \eqref{MF} holds independently for each $n \in \{ 0, 1, \dots , L \}$. The case $n=0$ is trivial and we are left with 
a total of $L$ non-linear equations describing the transfer matrix's spectrum. Each equation governs the eigenvalues in a particular
$\mathfrak{h}$-module although we are not using an appropriate label for $\Lambda$ as pointed out in Remark \ref{rema:lab}.
The determinantal condition \eqref{CC} for each value of $n$ encodes such non-linear functional equations in a compact way. In what follows
we inspect some of its particular cases.

\begin{eg}
Relation \eqref{CC} for $n = 1$ reads
\< \label{egn1}
\Lambda(x_0) \Lambda(x_1) &=& \Lambda(x_0) \left[ \phi_1 \frac{a(x_0 - x_1)}{b(x_0 - x_1)} \lambda_{\mathcal{A}}(x_1) + \phi_2 \frac{a(x_1 - x_0)}{b(x_1 - x_0)} \lambda_{\mathcal{D}}(x_1) \right] \nonumber \\
&& + \; \Lambda(x_1) \left[ \phi_1 \frac{a(x_1 - x_0)}{b(x_1 - x_0)} \lambda_{\mathcal{A}}(x_0) + \phi_2 \frac{a(x_0 - x_1)}{b(x_0 - x_1)} \lambda_{\mathcal{D}}(x_0) \right] \nonumber \\
&& - \; \left[ \frac{c(x_0 - x_1)}{b(x_0 - x_1)}  \right]^2 \left[ \phi_1 \lambda_{\mathcal{A}}(x_0) - \phi_2 \lambda_{\mathcal{D}}(x_0)  \right] \left[ \phi_1 \lambda_{\mathcal{A}}(x_1) - \phi_2 \lambda_{\mathcal{D}}(x_1)  \right] \; .  
\>
\end{eg}

\bigskip

\begin{eg}
The case $n=2$ is explicitly given by
\< \label{egn2}
\Lambda(x_0) \Lambda(x_1) \Lambda(x_2) &=& \mathfrak{m}_{2,2} \; \Lambda(x_0) \Lambda(x_1) + \mathfrak{m}_{1,1} \; \Lambda(x_0) \Lambda(x_2) + \mathfrak{m}_{0,0} \; \Lambda(x_1) \Lambda(x_2) \nonumber \\
&& + \; [ \mathfrak{m}_{1,2} \mathfrak{m}_{2,1} - \mathfrak{m}_{1,1} \mathfrak{m}_{2,2} ] \Lambda(x_0) + [ \mathfrak{m}_{0,2} \mathfrak{m}_{2,0} - \mathfrak{m}_{0,0} \mathfrak{m}_{2,2} ] \Lambda(x_1) \nonumber \\
&& + \; [ \mathfrak{m}_{0,1} \mathfrak{m}_{1,0} - \mathfrak{m}_{0,0} \mathfrak{m}_{1,1} ] \Lambda(x_2) + \mathfrak{M} \; . 
\>
As for the coefficients in \eqref{egn2} we have
\<
\mathfrak{m}_{i , j} &\coloneqq& \begin{cases}
\displaystyle \phi_1 \lambda_{\mathcal{A}} (x_i) \prod_{x \in X^0_i} \frac{a(x - x_i)}{b(x - x_i)} + \phi_2 \lambda_{\mathcal{D}} (x_i) \prod_{x \in X^0_i} \frac{a(x_i - x)}{b(x_i - x)} \qquad\qquad\qquad\qquad\qquad i = j \nonumber \\
\displaystyle \frac{c(x_j - x_i)}{b(x_j - x_i)} \left[ \phi_1 \lambda_{\mathcal{A}} (x_j) \prod_{x \in X^0_{i,j}} \frac{a(x - x_j)}{b(x - x_j)} - \phi_2 \lambda_{\mathcal{D}} (x_j) \prod_{x \in X^0_{i,j}} \frac{a(x_j - x)}{b(x_j - x)}   \right] \qquad\quad \mathrm{otherwise}
\end{cases} \nonumber \\
\>
while 
\<
\mathfrak{M} &\coloneqq& \left| \begin{matrix} \mathfrak{m}_{0,0} & \mathfrak{m}_{0,1} & \mathfrak{m}_{0,2} \cr  \mathfrak{m}_{1,0} & \mathfrak{m}_{1,1} & \mathfrak{m}_{1,2} \cr \mathfrak{m}_{2,0} & \mathfrak{m}_{2,1} & \mathfrak{m}_{2,2} \end{matrix} \right| \; .
\>

\end{eg}

%%%%%%%%%%%%%%%%%%%%%%%%%%%%%%%%%%%%%%%%%%
\section{Generating function of conserved quantities} \label{sec:QUANT}
%%%%%%%%%%%%%%%%%%%%%%%%%%%%%%%%%%%%%%%%%%%

The formulation of integrable differential equations as compatibility conditions of auxiliary linear problems has several practical consequences. 
One of them is a general and systematic procedure for constructing conserved quantities underlying the differential equation of interest. As for 
the auxiliary linear problem \eqref{MF}, one can also obtain conserved quantities in a systematic way. This is what we intend to
demonstrate in this section. 

\begin{defi}[Transport Function] \label{TF}
Let $\mathfrak{T}_{i \to j} \colon \C^{n+1} \to \C$ satisfy the property
\[ \label{tij}
\mathcal{F}_n (X_j^0) =  \mathfrak{T}_{i \to j} \; \mathcal{F}_n (X_i^0) \; .
\]
We refer to $\mathfrak{T}_{i \to j}$ as \emph{Transport Function} and it will be a key ingredient for constructing 
conserved quantities.
\end{defi}

\begin{thm} \label{FIF0a}
Let $1 \leq \alpha , \beta \leq n$ and write
\[ \label{VI}
\left( \mathcal{V}_i  \right)_{\alpha , \beta} \coloneqq \begin{cases}
- \mathrm{M}_{0, \alpha} \qquad\qquad \beta = i \cr
\mathrm{M}_{\beta, \alpha} \qquad\quad \text{otherwise}
\end{cases} 
\]
for matrix entries spanning the set $\{ \mathcal{V}_i \mid 0 \leq i \leq n \}$. As for $0 \leq i, j \leq n$, the transport function
associated with the auxiliary linear problem \eqref{MF} is	 given by
\[ \label{vivj}
\mathfrak{T}_{i \to j} = \frac{\mathrm{det}(\mathcal{V}_j)}{\mathrm{det}(\mathcal{V}_i)} \; .
\]
\end{thm}

\begin{proof}
Consider the system of linear equations \eqref{MFF} for $j = 1, 2, \dots , n$ and use Cramer's rule to solve $\mathcal{F}_n (X_i^0)$ in terms
of $\mathcal{F}_n (X)$. By doing so we find
\[ \label{FiF0}
\frac{\mathcal{F}_n (X_i^0)}{\mathcal{F}_n (X)} = \frac{\mathrm{det} (\mathcal{V}_i)}{\mathrm{det} (\mathcal{V}_0)} \; .
\]
Since \eqref{FiF0} is valid for $i \in \{1 ,2 , \dots , n \}$ we can readily write 
\[ \label{FiFj}
\frac{\mathcal{F}_n (X_i^0)}{\mathcal{F}_n (X_j^0)} = \frac{\mathrm{det} (\mathcal{V}_i)}{\mathrm{det} (\mathcal{V}_j)} 
\]
for $i, j \in \{ 1, 2, \dots , n \}$. This concludes our proof.
\end{proof}

The transport function defined by property \eqref{tij} exhibits several interesting properties. For instance, one can readily show it satisfies the composition 
property $\mathfrak{T}_{i \to k} = \mathfrak{T}_{j \to k} \mathfrak{T}_{i \to j}$.
Such property can be generalized to sequences $\mathcal{I}_m \coloneqq ( i_1, i_2 , \dots , i_m )$ as depicted in the following diagram
\[
\mathcal{F}_n (X_{i_1}^0) \xrightarrow{\mathfrak{T}_{i_1 \to i_2}} \mathcal{F}_n (X_{i_2}^0) \xrightarrow{\mathfrak{T}_{i_2 \to i_3}} \;\; \dots \;\; \mathcal{F}_n (X_{i_{m-1}}^0) \xrightarrow{\mathfrak{T}_{i_{m-1} \to i_m}} \mathcal{F}_n (X_{i_{m}}^0) \; . \nonumber 
\]

\medskip

\begin{lem}[Composition map] 
For sequences $\mathcal{I}_m$ one has the composition property
\[ \label{comp}
\mathfrak{T}_{i_1 \to i_m} = \mathop{\overleftarrow\prod}\limits_{1 \le l < m} \mathfrak{T}_{i_l \to i_{l+1}} \; .
\]
\end{lem}
\begin{proof}
The proof is trivial and follows from the repeated use of \eqref{tij}.
\end{proof}

\begin{rema}
The transport function for a closed sequence or loop  $(i_1, i_2 , \dots, i_{m-1}, i_1)$ is trivial and equals to $1$.
\end{rema}

\begin{rema}
Showing that the realization \eqref{vivj} for the auxiliary linear problem \eqref{MF} indeed satisfies \eqref{comp} is straightforward.
\end{rema}

\begin{thm}[Conserved quantities] \label{CQ}
Let $\partial_i \coloneqq \frac{\partial}{\partial x_i}$ denote the partial derivative with respect to $x_i$. Then the quantity 
\[ \label{gen}
\Theta_{i,j} \coloneqq \log{(\partial_i \log{(\mathfrak{T}_{i \to j})})}
\]
is a generating function of conserved quantities along the variable $x_j$. In other words, $\partial_j \Theta_{i,j} = 0$.
\end{thm}

\begin{proof}
Apply $\partial_j$ and $\partial_i$ to \eqref{tij} keeping in mind that $\partial_k f (X_k^0) = 0$ for arbitrary functions $f$.
By doing so we are left with two identities which can be manipulated to showing that $\partial_j \Theta_{i,j} = 0$.
\end{proof} 

The RHS of \eqref{gen} is given in terms of the set of variables $\{x_0 , x_1 , \dots , x_n \}$. However, it does not depend on $x_j$ on the solutions
manifold according to Theorem \ref{CQ}. In this way, we can conveniently write $\vec{r}_j \coloneqq (r_0 , r_1 , \dots , \hat{r_j}, \dots , r_n)$
for a $n$-tuple of integers and expand the RHS of \eqref{gen} in Taylor series as
\[ \label{expan}
\Theta_{i,j} = \sum_{\vec{r}_j \in \Z_{\geq 0}^n} \nabla_{\vec{r}_j}^{(i,j)} (x_j) \; \prod_{\substack{l = 0 \\ l \neq j}}^n x_l^{r_l} \; .
\]
Therefore, using Theorem \ref{CQ} we find $\partial_j \nabla_{\vec{r}_j}^{(i,j)} (x_j) = 0$ which allows us to conclude that the set 
$\{ \nabla_{\vec{r}}^{(i,j)} (x) \mid \vec{r} \in \Z_{\geq 0}^n \; ; \; 0 \leq i < j \leq n \}$ is a family of conserved quantities.

%%%%%%%%%%%%%%%%%%%%%%%%%%%%%%%%%%%%%%%%%%
\section{Solving the auxiliary linear problem} \label{sec:AUX}
%%%%%%%%%%%%%%%%%%%%%%%%%%%%%%%%%%%%%%%%%%%

In \Secref{sec:ALP} we have presented a rather general linear covering of non-linear functional equations mimicking the philosophy of the CISM.
An explicit realization of the proposed embedding was given in \Secref{sec:AFM} using the AF method.
In particular, we have shown that certain non-linear functional equations describing eigenvalues of the transfer matrix \eqref{tmat}
can be encoded in the compatibility condition \eqref{CC}. 
Examples of such non-linear equations have been made explicit in \eqref{egn1} and \eqref{egn2}. 
Hence, within our approach the transfer matrix's spectral problem corresponds to solving the aforementioned non-linear equations; which still seems to be 
a very complicated task. The latter is specially reinforced when inspecting the case \eqref{egn2}.
However, in this section we intend to show that this complication is deceptive and that the auxiliary linear problem \eqref{MF} can in fact assist us
in finding solutions.

\begin{lem} \label{FIF0}
Let $\mathcal{V}_i$ be matrices with entries defined in \eqref{VI}. Solutions of \eqref{MF} are then of the form 
$\mathcal{F}_n (X_i^0) = f_n \; \mathrm{det} (\mathcal{V}_i)$ with $f_n$ a function depending at most on the
set of variables  $\{x_0, x_1 , \dots , x_n \}$.
\end{lem}

\begin{proof}
The statement of Lemma \ref{FIF0} follows naturally from the fact that \eqref{tij} with \eqref{vivj} is factorized and valid for $0 \leq i, j \leq n$.
\end{proof}

In order to illustrate how Lemma \ref{FIF0} works in practice let us first examine its statements for the particular cases
$n = 0, 1, 2 , 3$ before addressing the general case.

\subsection{Case $n=0$} \label{sec:N0}
This case is trivial and we shall mainly use it to set up some conventions. In particular, Eq. \eqref{MF} for $n=0$ simply gives the condition
\[ \label{sig0}
\Sigma_0 \coloneqq \Lambda(x) - \lambda_{+} (x) = 0
\]
with 
\[
\lambda_{\pm} (x) \coloneqq \pm \phi_1 \lambda_{\mathcal{A}} (x) + \phi_2 \lambda_{\mathcal{D}} (x) \; . 
\]
Although \eqref{sig0} does not involve $\lambda_{-}$, we have defined it in such a way for later convenience.

\subsection{Case $n=1$} \label{sec:N1}
This is the first non-trivial case contained in Lemma \ref{FIF0}. Although it is a rather simple case one can already extract a large deal
of non-trivial information about the spectrum of the transfer matrix $\mathrm{T}$. We start our analysis by evaluating explicitly
the relevant determinants. In this way we find,
\< \label{v0v1}
\mathrm{det} (\mathcal{V}_0) &=& \frac{1}{b(x_0 - x_1)} \left[ \phi_1 a(x_0 - x_1) \lambda_{\mathcal{A}}(x_1) - \phi_2 a(x_1 - x_0) \lambda_{\mathcal{D}}(x_1) - b(x_0 - x_1) \Lambda(x_1) \right] \nonumber \\
\mathrm{det} (\mathcal{V}_1) &=&  \frac{c(x_0 - x_1)}{b(x_0 - x_1)} \left[ \phi_1  \lambda_{\mathcal{A}}(x_0) - \phi_2  \lambda_{\mathcal{D}}(x_0)  \right] \; . \nonumber \\
\>
and separation of variables allows us to conclude
\< \label{f1f0}
\mathcal{F}_1 (x_0) &=& g(x_0) \; \sinh{(\gamma)} \left[ \phi_1  \lambda_{\mathcal{A}}(x_0) - \phi_2  \lambda_{\mathcal{D}}(x_0)  \right] \nonumber \\
\mathcal{F}_1 (x_1) &=& g(x_0)  \left[ \phi_1 a(x_0 - x_1) \lambda_{\mathcal{A}}(x_1) - \phi_2 a(x_1 - x_0) \lambda_{\mathcal{D}}(x_1) - b(x_0 - x_1) \Lambda(x_1) \right] \; . \nonumber \\
\>
The precise identification with Lemma \ref{FIF0} is obtained with $f_1 = \frac{g(x_0)}{b(x_0 - x_1)}$.
In summary we have found two expressions for the same function $\mathcal{F}_1$. The consistency between these two formulae should then constraint
both functions $g$ and $\Lambda$.

From the first expression in \eqref{f1f0} we can see that $g$ and $\mathcal{F}_1$ only differ by a simple known factor. In its turn, the second
expression in \eqref{f1f0} tells us that the role played by $g$ is to remove the dependence of its RHS with the variable $x_0$. Moreover, the RHS
of \eqref{f1f0} contains explicitly the function $\Lambda$ and this allows one for instance to write the eigenvalue $\Lambda$ in terms of functions $g$ or 
$\mathcal{F}_1$. 

There are several possible ways of imposing consistency on the system of equations \eqref{f1f0} and the most obvious choice
is the identity 
\< \label{GG}
&& g(x_1) \; \sinh{(\gamma)} \left[ \phi_1  \lambda_{\mathcal{A}}(x_1) - \phi_2  \lambda_{\mathcal{D}}(x_1)  \right] = \nonumber \\
&& \qquad \quad \;\;\;	 g(x_0)  \left[ \phi_1 a(x_0 - x_1) \lambda_{\mathcal{A}}(x_1) - \phi_2 a(x_1 - x_0) \lambda_{\mathcal{D}}(x_1) - b(x_0 - x_1) \Lambda(x_1) \right] \; . \nonumber \\
\>
Eq. \eqref{GG} is automatically satisfied when $x_0 = x_1$ but it can be further analyzed using the derivative method. 
Here we shall adopt a different route though.

Write $\partial_i \coloneqq \frac{\partial}{\partial x_i}$ and use the condition $\partial_0 \mathcal{F}_1 (x_1) = 0$, with
$\mathcal{F}_1 (x_1)$ being given by the second expression of \eqref{f1f0}.
By doing so we are left with the relation
\<
\label{logg1}
\Lambda (x_1) &=& \left[ \partial_0 b(x_0 - x_1) + h(x_0) b(x_0 - x_1) \right]^{-1} \left\{ \phi_1 \partial_0 a(x_0 - x_1) \lambda_{\mathcal{A}}(x_1) \right. \nonumber \\
&&  - \; \left. \phi_2 \partial_0 a(x_1 - x_0) \lambda_{\mathcal{D}}(x_1) +  h(x_0) \left[ \phi_1 a(x_0 - x_1) \lambda_{\mathcal{A}}(x_1) - \phi_2 a(x_1 - x_0) \lambda_{\mathcal{D}}(x_1)  \right] \right\} \nonumber \\
\>
or equivalently
\< \label{logg0}
h(x_0) = - \frac{\left[ \phi_1 \partial_0 a(x_0 - x_1) \lambda_{\mathcal{A}}(x_1) - \phi_2 \partial_0 a(x_1 - x_0) \lambda_{\mathcal{D}}(x_1) - \partial_0 b(x_0 - x_1) \Lambda(x_1) \right]}{\left[ \phi_1 a(x_0 - x_1) \lambda_{\mathcal{A}}(x_1) - \phi_2 a(x_1 - x_0) \lambda_{\mathcal{D}}(x_1) - b(x_0 - x_1) \Lambda(x_1) \right]} \; , \nonumber \\
\>
upon some trivial rearrangements. In \eqref{logg1} and \eqref{logg0} we have introduced the function $h(x_0) \coloneqq  \partial_0 \log{ g(x_0)}$. The reason for 
writing both formulae \eqref{logg1} and \eqref{logg0} will become clearer in what follows. 

We proceed by implementing the condition $\partial_0 \Lambda(x_1) = 0$ using expression \eqref{logg1}. After some trivial simplifications we find that this condition
translates to the function $h$ as
\[ \label{hh}
\partial_0 h(x_0) + 1 = h(x_0)^2   \; .
\]
Eq. \eqref{hh} can be recognized as a Riccati equation and such type of equation admits a simple \emph{linearization}.
As far as Eq. \eqref{hh} is concerned, the linearization map is simply given by $h(x_0) = - \partial_0 \log{ \bar{g}(x_0)}$. Such map is closely related to the original definition of $h$ in
terms of the function $g$. The relation between $g$ and $\bar{g}$ simply reads $g(x_0) \bar{g}(x_0) = 1$.
Hence, in terms of the function $\bar{g}$, Eq. \eqref{hh} becomes $\left( \partial_0^2 - 1 \right) \bar{g}(x_0) = 0$ which is just a simple harmonic 
oscillator equation. For convenience we write its solution as
\[ \label{gg}
\bar{g}(x_0) = A_1 \sinh{(w_1 - x_0)} 
\]
which produces
\[ \label{hh0}
h(x_0) = \coth{(w_1 - x_0)} \; .
\]
The parameters $A_1$ and $w_1$ in \eqref{gg} and \eqref{hh0} are integration constants.
The substitution of \eqref{hh0} in \eqref{logg1} then gives us the following formula for the function $\Lambda$, 
\[ \label{sol1}
\Lambda (x_1) = \phi_1 \frac{a(w_1 - x_1)}{b(w_1 - x_1)} \lambda_{\mathcal{A}}(x_1) +  \phi_2 \frac{a(x_1 - w_1)}{b(x_1 - w_1)} \lambda_{\mathcal{D}}(x_1) \; .
\]
Now one can write down the function $g$ from \eqref{gg} and finally obtain the function $\mathcal{F}_1$. The latter can then be substituted back
in \eqref{MF} together with the obtained expression \eqref{sol1}. By doing so one finds that our auxiliary linear problem is indeed solved for arbitrary
integration constants $A_1$ and $w_1$. 

As for the derivation of \eqref{hh} we have employed the condition $\partial_0 \Lambda(x_1) = 0$. However, one could have instead used the condition
$\partial_1 h (x_0) = 0$ with the help of formula \eqref{logg0}. By doing so one finds the surface equation $\Sigma_1 = 0$ with differential function $\Sigma_1$
defined as
\[ \label{sig1}
\Sigma_1 \coloneqq \omega_0 (x) + \left[ \omega_1 (x)  - \Lambda (x) \right] \Sigma_0 - \sinh{(\gamma)} \lambda_{-} (x) \; \partial \Sigma_0 \;
\]
The functions $\omega_0$ and $\omega_1$ in \eqref{sig1} are then given by
\<
\omega_0 (x) &\coloneqq& \left[ \sinh{(\gamma)} \lambda_{-}(x) \right]^2 - \left[ (\cosh{(\gamma)} - 1) \lambda_{+}(x) \right]^2  \nonumber \\
&& + \; \sinh{(\gamma)} (\cosh{(\gamma)} - 1) \left[ \lambda_{-}(x) \partial \lambda_{+}(x) - \lambda_{+}(x) \partial \lambda_{-}(x) \right]  \nonumber \\
\omega_1 (x) &\coloneqq& \left[ 2 \cosh{(\gamma)} - 1 \right] \lambda_{+}(x) + \sinh{(\gamma)} \partial \lambda_{-} (x) \; . \nonumber \\
\>
In terms of $\Lambda (x)$ and its derivatives, equation $\Sigma_1 = 0$ explicitly reads
\[ \label{riccati}
- \sinh{(\gamma)} \lambda_{-} (x) \; \partial \Lambda(x) + J_1 (x) \Lambda(x) - \Lambda(x)^2 = J_0 (x) \; ,
\]
which is also a \emph{Riccati equation}. As for the coefficients $J_0$ and $J_1$ in \eqref{riccati} we have
\< \label{coefJ}
J_0 (x) &\coloneqq& \left[ \cosh{(\gamma)} \lambda_{+} (x)  \right]^2 - \left[ \sinh{(\gamma)} \lambda_{-} (x)  \right]^2 \nonumber \\
&& + \; \sinh{(\gamma)} \cosh{(\gamma)} \left[ \lambda_{+}(x) \partial \lambda_{-}(x) - \lambda_{-}(x) \partial \lambda_{+}(x)   \right] \nonumber \\
J_1 (x) &\coloneqq& 2 \cosh{(\gamma)} \lambda_{+}(x) + \sinh{(\gamma)} \partial \lambda_{-} (x)
\>

\begin{rema}
The very same equation \eqref{riccati} is found by simply taking the limit $x_0, x_1 \to x$ in \eqref{egn1}.
\end{rema}

\begin{rema}
The dependence of the coefficients \eqref{coefJ} with $x$ only appears through the functions $\lambda_{\pm}$.
\end{rema}

Riccati equations can be linearized for arbitrary coefficients and, in particular, the linearized form of \eqref{riccati}
can be straightforwardly integrated for arbitrary functions $\lambda_{+}$ and $\lambda_{-}$. As for \eqref{riccati} we consider the 
linearization map 
\[
\Lambda (x) = \sinh{(\gamma)} \lambda_{-} (x) \; \partial \log{(u(x))} \; ,
\]
which translates our Riccati equation into the following linear second-order ODE,
\< \label{uu}
\left[ \sinh{(\gamma)} \lambda_{-} (x) \right]^2 \partial^2 u(x) - \sinh{(\gamma)} \lambda_{-} (x) \left[ J_1 (x) - \sinh{(\gamma)} \partial \lambda_{-} (x) \right] \partial u(x) + J_0 (x) \;  u(x) = 0 \; . \nonumber \\
\>
Here we shall not discuss the resolution of \eqref{uu} in details but it is important to remark that, at the end of the day,
one finds the very same representation \eqref{sol1}.

\subsection{Case $n=2$}  \label{sec:N2}

We proceed with the analysis of Lemma \ref{FIF0} for the case $n=2$. For that it is convenient to introduce matrices
$\widetilde{\mathcal{V}}_1$ and $\widetilde{\mathcal{V}}_2$ satisfying the conditions
\<
\mathrm{det} (\mathcal{V}_1) &=& \frac{c \; \mathrm{det} (\widetilde{\mathcal{V}}_1)}{b(x_0 - x_1) b(x_0 - x_2) b(x_2 - x_0)} \nonumber \\
\mathrm{det} (\mathcal{V}_2) &=& \frac{c \; \mathrm{det} (\widetilde{\mathcal{V}}_2)}{b(x_0 - x_1) b(x_0 - x_2) b(x_1 - x_0)} \; .
\>
Using elementary determinant-preserving operations we can define such matrices $\widetilde{\mathcal{V}}_1$ and $\widetilde{\mathcal{V}}_2$
as
\< \label{Vbar}
\left( \widetilde{\mathcal{V}}_1 \right)_{\alpha , \beta} &\coloneqq& \begin{cases}
- c^{-1} b(x_0 - x_1) b(x_0 - x_2) \left( \mathcal{V}_1 \right)_{\alpha, 1} \qquad \beta = 1 \cr
b(x_0 - x_2) \left( \mathcal{V}_1 \right)_{\alpha, 2} \qquad \; \qquad \qquad\qquad  \beta = 2 
\end{cases} \nonumber \\
\left( \widetilde{\mathcal{V}}_2 \right)_{\alpha , \beta} &\coloneqq& \begin{cases}
b(x_0 - x_1) \left( \mathcal{V}_2 \right)_{\alpha, 1} \qquad \; \qquad \qquad\qquad  \beta = 1 \cr 
- c^{-1} b(x_0 - x_1) b(x_0 - x_2) \left( \mathcal{V}_2 \right)_{\alpha, 2} \qquad \beta = 2 
\end{cases}  \; .
\>
Hence, as described in Lemma \ref{FIF0}, we have
\[
\frac{\mathcal{F}_2 (x_0 , x_2)}{\mathcal{F}_2 (x_0 , x_1)} = \frac{\mathrm{det} (\mathcal{V}_1)}{\mathrm{det} (\mathcal{V}_2)} = \frac{b(x_0 - x_1)}{b(x_0 - x_2)} \frac{\mathrm{det} (\widetilde{\mathcal{V}}_1)}{\mathrm{det} (\widetilde{\mathcal{V}}_2)} \; .
\]
Moreover, now using row/column transpositions, one can readily verify from \eqref{Vbar} that $\mathrm{det} (\widetilde{\mathcal{V}}_1) = \pi_{1,2} \; \mathrm{det} (\widetilde{\mathcal{V}}_2)$. In their turn, the explicit evaluation of 
the aforementioned determinants reveals that $\mathrm{det} (\widetilde{\mathcal{V}}_1) \in \C [x_0^{\pm 1} , x_2^{\pm 1}]$ and $\mathrm{det} (\widetilde{\mathcal{V}}_2) \in \C [x_0^{\pm 1} , x_1^{\pm 1}]$.
Such properties then allow us to conclude that
\[ \label{FDET1}
\mathcal{F}_2 (x_0 , x_2) = g \frac{\mathrm{det} (\widetilde{\mathcal{V}}_1)}{b(x_0 - x_2)} 
\]
and 
\[ \label{FDET2}
\mathcal{F}_2 (x_0 , x_1) = g \frac{\mathrm{det} (\widetilde{\mathcal{V}}_2)}{b(x_0 - x_1)} 
\]
for a given function $g$ depending at most on $x_0$, $x_1$ and $x_2$.  
In order to make contact with Lemma \ref{FIF0} we have $f_2 = g \; c^{-1} b(x_0 - x_1) b(x_2 - x_0)$.

\begin{prop} \label{G0}
The function $g$ depends only on $x_0$.
\end{prop}
\begin{proof}
Since $\C [x_0^{\pm 1} , x_2^{\pm 1}] \ni \frac{\mathrm{det} (\widetilde{\mathcal{V}}_1)}{b(x_0 - x_2)}$, the function $g$ can not depend on $x_1$ as this would
contradict the LHS of \eqref{FDET1}. On the other hand, $\C [x_0^{\pm 1} , x_1^{\pm 1}] \ni \frac{\mathrm{det} (\widetilde{\mathcal{V}}_2)}{b(x_0 - x_1)}$
and a function $g$ depending on $x_2$ would contradict \eqref{FDET2}. Thus $g = g(x_0)$.
\end{proof}

Next we look at the function $\mathcal{F}_2 (x_1 , x_2)$ according to Lemma \ref{FIF0}. In this way we have
\[
\mathcal{F}_2 (x_1 , x_2) =  g(x_0) \; c^{-1} b(x_0 - x_1) b(x_2 - x_0) \; \mathrm{det} (\mathcal{V}_0)  \; ,
\]
and we can use this representation to determine the function $g(x_0)$. For that we consider the following sequence of steps:
\medskip
\par{\bf{Step $1$.}} Use the identity $\partial_0 \mathcal{F}_2 (x_1 , x_2) = 0$ to find an expression for $\Lambda (x_1)$ in terms of $\Lambda(x_2)$
and $h(x_0) \coloneqq \partial_0 \mathrm{log} g(x_0)$.
\medskip
\par{\bf{Step $2$.}} Use expression obtained for $\Lambda (x_1)$ and implement the condition $\partial_0 \Lambda(x_1) = 0$ to write $\Lambda (x_2)$
in terms of the function $h$ and its derivatives.
\medskip
\par{\bf{Step $3$.}} At last, we consider the identity $\partial_0 \Lambda(x_2) = 0$ using $\Lambda(x_2)$ obtained in \emph{Step 2}. This yields the following
ODE for the function $h$,
\[ \label{ric_n2}
\partial^2 h - 3 h \partial h + h (h^2 - 4)  = 0 \; ,
\]
where $h = h(x)$ and $\partial \coloneqq \frac{\dd}{\dd x}$.

Eq. \eqref{ric_n2} is a second-order non-linear ODE which can be regarded as a higher-order Riccati equation since it admits the same
kind of linearization map employed for Eq. \eqref{hh}. Then using the map $h(x) = - \partial \log{ \bar{g}(x)}$ we find that \eqref{ric_n2} turns into the
simple equation $\left( \partial^3 - 4 \partial \right) \bar{g}(x) = 0$. Although the latter is a third-order linear ODE, its structure allows one
to regard it as a second-order equation for the function $\partial \bar{g}(x)$. For convenience we write the solution of this equation as
\[ \label{gg_2}
\bar{g}(x) = A_2 \sinh{(w_1 - x)} \sinh{(w_2 - x)} \; ,
\]
which straightforwardly leads to 
\[ \label{h_2}
h(x) = \coth{(w_1 - x)} + \coth{(w_2 - x)} \; .
\]
In \eqref{gg_2} we have three integration constants, namely $A_2$, $w_1$ and $w_2$, which is a consequence of having $\bar{g}$ described by a third-order
linear ODE. Due to the particular form of our linearization map only two integration constants are transferred to formula \eqref{h_2}.
We can now replace solution \eqref{h_2} in the expression for $\Lambda$ found in \emph{Step 2}. By doing so we are left with the representation
\[ \label{sol2}
\Lambda (x) = \phi_1 \frac{a(w_1 - x)}{b(w_1 - x)} \frac{a(w_2 - x)}{b(w_2 - x)} \lambda_{\mathcal{A}}(x) +  \phi_2 \frac{a(x - w_1)}{b(x - w_1)} \frac{a(x - w_2)}{b(x - w_2)} \lambda_{\mathcal{D}}(x) \; .
\]

\begin{rema}
We have used repeated notation for the cases $n=1$ and $n=2$. In particular, for the functions $g$, $h$ and $\Lambda$, but we hope 
the context makes the distinction clear.
\end{rema}

So far in \Secref{sec:N2} we have taken the route of deriving differential equations for the function $h$. As a matter of fact this route yields a
much simpler derivation of the function $\Lambda$. Nevertheless, we can still write differential equations satisfied by the eigenvalue $\Lambda$ itself.
For that we consider the following alternative steps.
\medskip
\par{\bf{Step $\bar{1}$.}} Use the identity $\partial_0 \mathcal{F}_2 (x_1 , x_2) = 0$ to find an expression for $h(x_0) \coloneqq \partial_0 \mathrm{log} g(x_0)$
in terms of $\Lambda(x_1)$ and $\Lambda(x_2)$.
\medskip
\par{\bf{Step $\bar{2}$.}} Using the expression obtained in \emph{Step $\bar{1}$} we look at the condition $\partial_2 h(x_0) = 0$. This allows us to write
$\Lambda(x_1)$ in terms of $\Lambda(x_2)$ and $\partial_2 \Lambda(x_2)$.
\medskip
\par{\bf{Step $\bar{3}$.}} Next we substitute the expression for $\Lambda(x_1)$ obtained in \emph{Step $\bar{2}$} in the condition $\partial_2 \Lambda(x_1) = 0$.
This latter step leaves us with the non-linear ODE $\Sigma_2 = 0$.
\medskip

In its turn the differential function $\Sigma_2$ is given by
\<
\label{sig2}
\Sigma_2 &\coloneqq& \frac{(q-1)}{(q+1)} \left[ \xi_{0}(x) - 2 q \xi_1 (x) \Lambda (x)  \right] + 2 q (q -1)^2 \left[ \xi_2 (x) - 4 q \xi_3 (x) \Lambda(x)  \right] \partial \Lambda(x) \nonumber \\
&& + \; \frac{8 q^3}{(q^2 -1)} \left[ \xi_4 (x) - 8 q^2 \lambda_{-} (x) \Lambda(x) + 4 q (q^2 -1) \lambda_{+} (x)  \partial \Lambda(x)  \right] \Sigma_1 \nonumber \\
&& - \; 4 q^2 \left[ (q^2 +1) \xi_5 (x) + 4 q^2 \lambda_{+}(x)  \Lambda(x)  \right] \partial \Sigma_1 \nonumber \\
\>
with auxiliary functions $\xi_i (x)$ defined explicitly in \Appref{app:FUN}. In formula \eqref{sig2} we have defined $\Sigma_2$ in terms of the
differential function $\Sigma_1$ defined in \eqref{sig1}. The explicit substitution of \eqref{sig1} in \eqref{sig2} yields a non-linear differential equation for 
$\Lambda(x)$ which can be regarded as a higher-order version of the Riccati equation. However, it is important to remark here that \eqref{sig2}
is not the simplest differential equation satisfied by $\Lambda$ in the $n=2$ $\mathfrak{h}$-module. We shall return to this problem in 
\Secref{sec:6VC} where we also find a standard Riccati equation describing the eigenvalue $\Lambda$ in the sector $n=2$.

\subsection{Case $n=3$}  \label{sec:N3}

Our analysis of Lemma \ref{FIF0} for $n=3$ consists of a straightforward generalization of the $n=2$ case. For convenience we first introduce matrices
$\widetilde{\mathcal{V}}_1$, $\widetilde{\mathcal{V}}_2$ and $\widetilde{\mathcal{V}}_3$ fulfilling the condition
\[ \label{VbV}
\mathrm{det} \left( \mathcal{V}_i \right) = \frac{c \; b(x_0 - x_i)}{\prod_{j=1}^3 b(x_0 - x_j)^2} \mathrm{det} ( \widetilde{\mathcal{V}}_i ) \; .
\]
Although \eqref{VbV} seems to be an arbitrary condition, it is strongly supported by formulae \eqref{MI} and \eqref{VI}. In this way, we find
appropriate matrices $\widetilde{\mathcal{V}}_i$ with entries defined as
\<
( \widetilde{\mathcal{V}}_i )_{\alpha , \beta} \coloneqq \begin{cases}
( \mathcal{V}_i )_{\alpha , i} \; \displaystyle \prod_{j=1}^3 b(x_0 - x_j)  \qquad \qquad \beta = i \cr
( \mathcal{V}_i )_{\alpha , \beta} \; b(x_0 - x_{\beta}) \qquad \qquad \mathrm{otherwise}
\end{cases} \; .
\>
Moreover, using simple row/column transpositions one can readily show that  $\mathrm{det} ( \widetilde{\mathcal{V}}_i ) = \pi_{i,j} \mathrm{det} ( \widetilde{\mathcal{V}}_j )$.
Next we carefully inspect such determinants and the explicit evaluation of $\mathrm{det} ( \widetilde{\mathcal{V}}_1 )$ reveals that $\mathrm{det} ( \widetilde{\mathcal{V}}_1 ) \in \C [ x_0^{\pm 1} , x_2^{\pm 1} , x_3^{\pm 1}]$.
Since $\mathrm{det} ( \widetilde{\mathcal{V}}_i ) = \pi_{i,j} \mathrm{det} ( \widetilde{\mathcal{V}}_j )$ we can also conclude that 
$\mathrm{det} ( \widetilde{\mathcal{V}}_2 ) \in \C [ x_0^{\pm 1} , x_1^{\pm 1} , x_3^{\pm 1}]$ and $\mathrm{det} ( \widetilde{\mathcal{V}}_3 ) \in \C [ x_0^{\pm 1} , x_1^{\pm 1} , x_2^{\pm 1}]$.

Now using Lemma \ref{FIF0} we are left with the relation
\[ \label{fifj3}
\frac{\mathcal{F}_3 (X_i^0)}{\mathcal{F}_3 (X_j^0)} = \frac{b(x_0 - x_i)}{b(x_0 - x_j)} \frac{\mathrm{det} ( \widetilde{\mathcal{V}}_i )}{\mathrm{det} ( \widetilde{\mathcal{V}}_j )} \qquad \quad  i, j \in \{1 ,2 ,3 \} \; .
\]
In particular, since relation \eqref{fifj3} is fulfilled for $i, j \in \{1 ,2 ,3 \}$, one can readily conclude that
\[ \label{pf3}
\mathcal{F}_3 (X_i^0) = \frac{g}{\displaystyle \prod_{\substack{j=1 \\ j \neq i}}^3 b(x_0 - x_j)} \mathrm{det} ( \widetilde{\mathcal{V}}_i ) 
\]
for $i \in \{1 , 2 , 3 \}$ and a given function $g$ depending at most on $x_0$, $x_1$, $x_2$ and $x_3$. However, 
one can prove $g = g(x_0)$ along the same lines of Proposition \ref{G0}. Hence, we find no need to repeat such proof here.

Lemma \ref{FIF0} still contains more information. For instance, gathering our results so far we can also write
\[ \label{f3}
\mathcal{F}_3 (X) = c^{-1} g(x_0) \; \mathrm{det} (\mathcal{V}_0) \prod_{j=1}^3 b(x_0 - x_j) \; .
\]
In this way the complete determination of the multivariate function $\mathcal{F}_3$ is reduced to fixing the univariate functions
$g(x_0)$ and $\Lambda(x_i)$. The latter can be achieved through the following sequence of steps.

\medskip
\par{\bf{Step $1$.}} Substitute representation \eqref{f3} in the identify $\partial_0 \mathcal{F}_3 (X) = 0$ and use it to write
$\Lambda (x_1)$ in terms of $\Lambda (x_2)$, $\Lambda (x_3)$ and $h(x_0) \coloneqq \partial_0 \log{g(x_0)}$.

\medskip
\par{\bf{Step $2$.}} Implement the condition $\partial_0 \Lambda(x_1) = 0$ using the expression for $\Lambda(x_1)$ obtained in the previous
step. By doing so one finds an expression for $\Lambda (x_2)$ in terms of $\Lambda (x_3)$, $h(x_0)$ and its first derivative. 

\medskip
\par{\bf{Step $3$.}} Using the expression for $\Lambda (x_2)$ obtained in \emph{Step $2$} we consider the condition $\partial_0 \Lambda(x_2) = 0$.
This allows us to write $\Lambda (x_3)$ in terms of $h(x_0)$ and derivatives. 

\medskip
\par{\bf{Step $4$.}} Insert representation for $\Lambda (x_3)$ obtained in \emph{Step $3$} into the identity $\partial_0 \Lambda(x_3) = 0$. After simplifications
we uncover the following differential equation for the function $h$,
\[ \label{h3}
\partial^3 h - 4 h \partial^2 h -  \left( 10 - 6 h^2 + 3 \partial h \right) \partial h - (h^2 - 1)(h^2 - 9) = 0 \; .
\]

Although Eq. \eqref{h3} is a third order non-linear differential equation, it can be exactly solved using the same procedure described in 
\Secref{sec:N1} and \Secref{sec:N2}. For that we use again the map $h(x) = - \partial \log{\bar{g}(x)}$ which transforms 
\eqref{h3} into the linear equation $( \partial^2 - 1 ) ( \partial^2 - 9 ) \; \bar{g}(x) = 0$. The latter is a fourth order linear ODE
with solution given by
\[ \label{g3}
\bar{g} (x) = A_3 \sinh{(w1 - x)} \sinh{(w2 - x)} \sinh{(w3 - x)} \; .
\]
Solution \eqref{g3} contains four arbitrary integration constants, namely $A_3$, $w1$, $w2$ and $w3$; and by reversing our linearization
map  we find 
\[ \label{solh3}
h(x) = \coth{(w1 - x)} + \coth{(w2 - x)} + \coth{(w3 - x)} \; .
\]
The function $g$ entering formulae \eqref{pf3} and \eqref{f3} then follows straightforwardly from \eqref{solh3}. As far as the full determination
of the function $\mathcal{F}_3$ is concerned, we still need to fix the function $\Lambda$. For that we use the expression for $\Lambda$
obtained in \emph{Step $3$} which is given in terms of $h$ and its derivatives. Hence, by substituting \eqref{solh3} in the expression for 
$\Lambda$, we find at last
\< \label{sol3}
\Lambda (x) = \phi_1 \frac{a(w_1 - x)}{b(w_1 - x)} \frac{a(w_2 - x)}{b(w_2 - x)} \frac{a(w_3 - x)}{b(w_3 - x)}  \lambda_{\mathcal{A}}(x) +  \phi_2 \frac{a(x - w_1)}{b(x - w_1)} \frac{a(x - w_2)}{b(x - w_2)} \frac{a(x - w_3)}{b(x - w_3)}  \lambda_{\mathcal{D}}(x) \; . \nonumber \\
\>

In Sections \ref{sec:N1} and \ref{sec:N2} we have also presented non-linear differential equations satisfied by the eigenvalue $\Lambda$ itself.
They consist of (higher-order) Riccati equations but with rather involved expressions for their coefficients. 
Interestingly, the non-linear differential equations satisfied by the function $h$ are also of Riccati type, albeit with $\Z$-valued
coefficients. Here we shall not present differential equations for $\Lambda$ in the $n=3$ case but it is worth remarking that such equation
can be obtained using the alternative steps discussed in \Secref{sec:N2}.

\subsection{General case} \label{sec:NN}
In Sections \ref{sec:N0} through \ref{sec:N3} we have used a concrete example to illustrate how the auxiliary linear problem \eqref{aux}
can employed for solving non-linear functional equations encoded in Lemma \ref{det}. The functional equations we are considering describe 
the spectrum of the six-vertex model's transfer matrix and explicit examples are given by \eqref{egn1} and \eqref{egn2}.
The latter examples describe eigenvalues $\Lambda$ in the $n=1$ and $n=2$ $\mathfrak{h}$-modules respectively and a detailed analysis of those 
cases can be found in the previous subsections. 
As for such cases our analysis also reveals a close relation between our functional relations and Riccati non-linear ODEs.
Our approach has been implemented on a case-by-case basis and it is fair to say that an unified analysis for arbitrary $\mathfrak{h}$-modules
is still beyond our reach. Nevertheless, the results previously presented exhibit a clear pattern suggesting the following generalization for $n > 3$.

Let us introduce the set of matrices $\{ \widetilde{\mathcal{V}}_i \mid 1 \leq i \leq n \}$ with entries defined as
\< \label{tildeV}
( \widetilde{\mathcal{V}}_i )_{\alpha , \beta} \coloneqq \begin{cases}
( \mathcal{V}_i )_{\alpha , i} \; \displaystyle \prod_{j=1}^n b(x_0 - x_j)  \qquad \qquad \beta = i \cr
( \mathcal{V}_i )_{\alpha , \beta} \; b(x_0 - x_{\beta}) \qquad \qquad \mathrm{otherwise}
\end{cases} \; .
\>
From \eqref{tildeV} one can readily demonstrate the properties
\[
\mathrm{det} ( \widetilde{\mathcal{V}}_i ) = \pi_{i,j} \; \mathrm{det} ( \widetilde{\mathcal{V}}_j ) \quad \mbox{and} \quad \mathrm{det} \left( \mathcal{V}_i \right) = \frac{c \; b(x_0 - x_i)}{\prod_{j=1}^n b(x_0 - x_j)^2} \mathrm{det} ( \widetilde{\mathcal{V}}_i ) \; .
\]
In this way Lemma \ref{FIF0} leaves us with relations
\[ \label{fifjn}
\frac{\mathcal{F}_n (X_i^0)}{\mathcal{F}_n (X_j^0)} = \frac{b(x_0 - x_i)}{b(x_0 - x_j)} \frac{\mathrm{det} ( \widetilde{\mathcal{V}}_i )}{\mathrm{det} ( \widetilde{\mathcal{V}}_j )} \qquad \qquad \quad i, j \in \{1 ,2 , \dots , n \} \; .
\]

\begin{conjec} \label{jura} Let $\widetilde{\mathcal{V}}_i$ be a matrix with entries defined in \eqref{tildeV}. Then
\[
\mathrm{det} ( \widetilde{\mathcal{V}}_i ) \in \C [ x_0^{\pm 1} , x_1^{\pm 1}, \dots, \widehat{x_i}^{\pm 1}, \dots , x_n^{\pm 1}] \; .
\]
\end{conjec}
\medskip
Now equations \eqref{fifjn} together with Conjecture \ref{jura} implies in
\[ 
\mathcal{F}_n (X_i^0) = \frac{g(x_0)}{\displaystyle \prod_{\substack{j=1 \\ j \neq i}}^n b(x_0 - x_j)} \mathrm{det} ( \widetilde{\mathcal{V}}_i ) 
\]
for $i \in \{ 1, 2, \dots , n \}$. Moreover, from Lemma \ref{FIF0} we additionally have
\[ \label{fn}
\mathcal{F}_n (X) = c^{-1} g(x_0) \; \mathrm{det} (\mathcal{V}_0) \prod_{j=1}^n b(x_0 - x_j) \; .
\]
Hence, in order to determine solutions of the auxiliary linear problem, namely $\mathcal{F}_n$; we only need to fix functions  
$g$ and $\Lambda$. The latter can be obtained through the following sequence of recursive steps.

\medskip
\par{\bf{Step $1$.}} Substitute \eqref{fn} in the identity $\partial_0 \mathcal{F}_n (X) = 0$ and solve the resulting expression
for $\Lambda(x_1)$. This yields a formula for $\Lambda(x_1)$ in terms of $\Lambda(x_2)$, $\Lambda(x_3)$, $\dots$, $\Lambda(x_n)$ 
and $h(x_0) \coloneqq \partial_0 \log{g(x_0)}$.

\medskip
\par{\bf{Step $i$ ($2 \leq i \leq n$).}} Write expression for $\Lambda(x_i)$ using the identify  $\partial_0 \Lambda(x_{i-1}) = 0$ with
function $\Lambda(x_{i-1})$ obtained in \emph{Step $i-1$}.

\medskip
\par{\bf{Step $n+1$.}} Use expression for $\Lambda(x_n)$ obtained in \emph{Step $n$} to evaluate the trivial identity $\partial_0 \Lambda(x_{n}) = 0$.
This last step leaves one with a non-linear ODE for the function $h$.

The above described procedure involves a sequence of $n+1$ steps and, by executing it, we obtain a non-linear ODE for the function $h$. 
Here, however, we can only conjecture the explicit form of such equation for $n > 3$. For that we introduce the differential function $\Upsilon_n$ reading
\<
\Upsilon_n \coloneqq \begin{cases}
\displaystyle \partial \prod_{l=1}^{\frac{n}{2}} \left[ \partial^2 - (2 l)^2  \right] \qquad \quad\quad n \in 2 \Z \cr
\displaystyle \prod_{l=1}^{\frac{n+1}{2}} \left[ \partial^2 - (2 l - 1 )^2  \right] \qquad \quad n \in 2 \Z + 1
\end{cases} \; .
\>
In this way the aforementioned differential equation for $h$ reads 
\[ \label{ricN}
\Upsilon_n \; e^{- \int h(x) \dd x} = 0 \; ,
\]
with integration symbol being merely formal and simply used to denote the anti-derivative. The particular cases of \eqref{ricN}
worked out in Sections \ref{sec:N1}, \ref{sec:N2} and \ref{sec:N3} show the connection with Riccati type equations.
Moreover, Eq. \eqref{ricN} can be readily linearized through the map $h(x) = - \partial \log{\bar{g}(x)}$ and the solution of the resulting
equation is given by
\[
\bar{g}(x) = A_n \prod_{l=1}^n \sinh{(w_l - x)} \; .
\]
The corresponding function $h$ then reads
\[ \label{hn}
h(x) = \sum_{l=1}^n \coth{(w_l - x)} 
\]
from which one can readily write down the function $g$. As for the full determination of $\mathcal{F}_n$ we still need to fix the function
$\Lambda$. For convenience we then use expression for $\Lambda(x_n)$ obtained in \emph{Step $n$} and substitute formula \eqref{hn}. 
By doing so we are left with the following representation,
\[ \label{lambdaN}
\Lambda (x) = \phi_1 \prod_{l=1}^n \frac{a(w_l - x)}{b(w_l - x)} \lambda_{\mathcal{A}}(x) +  \phi_2 \prod_{l=1}^n \frac{a(x - w_l)}{b(x - w_l)}  \lambda_{\mathcal{D}}(x) \; .
\]

Some comments are in order at this stage. For instance, due to the definition of the transfer matrix \eqref{tmat}, its eigenvalues are a function of
the spectral parameter $x \in \C$. They also carry fixed complex parameters $\gamma$, $\mu_i$, $\phi_1$ and $\phi_2$. However, expression
\eqref{lambdaN} contains $n$ additional parameters, namely the set $\{ w_l \mid 1 \leq l \leq n \}$, which would then produce a continuous 
spectrum for \eqref{tmat}. That is clearly not the case since \eqref{tmat} is a matrix of dimension $2^L \times 2^L$. In fact, when looking to the 
transfer matrix's spectrum as solutions of functional/differential equations, one still needs to formulate it as an \emph{Initial or Boundary Value Problem}.
This will be discussed in \Secref{sec:DISCR}.

\subsection{Functions $\mathcal{F}_n$ as scalar products} \label{sec:scp}

As for the six-vertex model realization of the auxiliary linear problem \eqref{MF}, we can give a meaning to the functions $\mathcal{F}_n$.
For instance, from Theorem \ref{6V} we can see that $\mathcal{F}_n$ is the projection of a transfer matrix eigenvector $\ket{\Psi}$ on particular
vectors. However, our results are basis-independent in the sense that we have not chosen a particular representation for the eigenvector
$\ket{\Psi}$. In contrast to our approach, representations of such eigenvectors are fixed from the very beginning within the algebraic
Bethe ansatz method. In case one would like to link both approaches, this can be accomplished by simply setting 
\[ \label{aba}
\bra{\Psi} = \bra{0} \mathcal{C}(w_n) \; \mathcal{C}(w_{n-1}) \dots \mathcal{C}(w_1) 
\]
with parameters $w_i$ properly adjusted. Hence, by taking into account representation \eqref{aba}, one can precisely identify $\mathcal{F}_n$ with 
scalar products of Bethe vectors firstly studied in \cite{Korepin_1982} and \cite{Slavnov_1989}. In particular, the results presented here then make
contact with the continuous determinantal representations recently obtained in \cite{Galleas_2016c}.

%%%%%%%%%%%%%%%%%%%%%%%%%%%%%%%%%%%%%%%%%%
\section{Non-linear PDEs} \label{sec:PDE}
%%%%%%%%%%%%%%%%%%%%%%%%%%%%%%%%%%%%%%%%%%%

In \Secref{sec:AUX} we have obtained non-linear ODEs from the analysis of the auxiliary linear problem \eqref{MI}. More precisely, in this section
we refer to equations \eqref{ricN} which can be regarded as higher-order analogues of Riccati equation.
Solutions $h$ can be neatly written as \eqref{hn} and they exhibit interesting features. For instance, let $\chi$ and $\tau$ denote \emph{spatial}
and \emph{temporal} coordinates respectively. Hence, by writing $x \eqqcolon \chi - \omega \tau$ one can regard $h$ as a superposition of travelling-waves
of equal amplitude and same velocity $\omega$. The \emph{phase} of such waves differs and they are characterized by the parameters $w_i$. 
In this way one can ask if there exist PDEs corresponding to the ODEs \eqref{ricN} within the travelling-wave ansatz. This relation between PDEs
and ODEs can also be found in the KdV equation and we shall describe it in what follows for illustrative purposes.

Let $f = f(\chi, \tau)$ be a solution of the 
KdV equation
\[ \label{kdv}
\partial_{\tau} f + 6 f \partial_{\chi} f + \partial_{\chi}^3 f = 0
\]
with $\partial_{\tau} \coloneqq \frac{\partial}{\partial \tau}$ and $\partial_{\chi} \coloneqq \frac{\partial}{\partial \chi}$.
Hence, using the travelling-wave ansatz $f(\chi , \tau) = \bar{f}(\chi - \omega \tau)$ we find that \eqref{kdv} reduces to
$ - \omega \partial \bar{f} + 6 \bar{f} \partial \bar{f} + \partial^3 \bar{f} = 0$. The latter can be integrated with respect to $x$ leaving us
with the following second-order equation,
\[ \label{rkdv}
- \omega \bar{f} + 3 \bar{f}^2 + \partial^2 \bar{f} = \theta \; .
\]
The parameter $\theta$ is an integration constant and we thus have a direct relation between solutions of \eqref{rkdv} and travelling-wave
solutions of \eqref{kdv} with particular boundary conditions. 

Now we would like to associate non-linear PDEs with equations \eqref{ricN} in the same manner as one has for the KdV equation. We then proceed
on a case-by-case basis keeping in mind that the connection with the six-vertex model for $n>3$ only holds up to Conjecture \ref{jura}.

\begin{thm}[Case $n=1$] \label{PDE_1}
Eq. \eqref{hh} describes particular travelling-wave solutions of the non-linear PDE 
\[ \label{pde_1}
\partial_{\tau}^2 \psi -  \omega^2  \partial_{\chi} \left( \psi^2 \right) = 0
\]
with velocity $\omega$.
\end{thm}
\begin{proof}
Identify $\psi (\chi, \tau) = h(\chi - \omega \tau)$ in \eqref{pde_1} and integrate the resulting equation with respect to
$x = \chi - \omega \tau$ to find an ODE which can specialized to \eqref{hh}.
\end{proof}

\begin{thm}[Case $n=2$] \label{PDE_2}
Let $\psi = \psi (\chi, \tau)$ satisfy the following non-linear PDE,
\[ \label{pde_2}
\partial_{\tau}^3 \psi + \frac{3}{2} \omega^3 \partial_{\chi}^2 \left( \psi^2 \right) - \omega^3 \partial_{\chi} \left[ \psi (\psi^2 - 4)  \right] = 0 \; .
\]
Then particular travelling-wave solutions of velocity $\omega$ are described by Eq. \eqref{ric_n2}.
\end{thm}
\begin{proof}
Analogous to the proof of Theorem \ref{PDE_1} but using \eqref{ric_n2} instead of \eqref{hh}.
\end{proof}

\begin{thm}[Case $n=3$]
Eq. \eqref{h3} corresponds to the reduced form of the non-linear PDE
\< \label{pde_3}
&& \omega^{-4} \partial_{\tau}^4 \psi +  \partial_{\chi} \left[ \psi^2 (10 - \psi^2) + (\partial_{\chi} \psi)^2  \right] \nonumber \\
&& \;\quad + \; 2 \partial_{\chi}^2 \left[ \psi (\psi^2 - 5) \right] - 2 \partial_{\chi}^3 \left( \psi^2 \right) = 0 
\>
capturing solutions $\psi(\chi, \tau) = h(\chi - \omega \tau)$.
\end{thm}
\begin{proof}
Same proof as for Theorems \ref{PDE_1} and \ref{PDE_2}.
\end{proof}

\begin{rema}
Explicit solutions of \eqref{ricN} are given by \eqref{hn} and $h(\chi - \omega \tau)$ solves \eqref{pde_1}, \eqref{pde_2} and \eqref{pde_3}
with the respective value of $n$.
\end{rema}

%%%%%%%%%%%%%%%%%%%%%%%%%%%%%%%%%%%%%%%%%%
\section{Conserved quantities in the six-vertex model} \label{sec:6VC}
%%%%%%%%%%%%%%%%%%%%%%%%%%%%%%%%%%%%%%%%%%%

In \Secref{sec:QUANT} we have presented a generating function of conserved quantities associated with the auxiliary linear problem \eqref{MF}.
A realization of such linear problem describing the six-vertex model eigenvalue problem was then given in \eqref{MI}. The main result 
of \Secref{sec:QUANT} is Theorem \ref{CQ} which expresses the generating function $\Theta_{i,j}$ in terms of transport functions 
$\mathfrak{T}_{i \to j}$ defined by \eqref{tij}. In particular, although \eqref{gen} is rather simple and compact, it is worth remarking
that such formula is not limited to the eigenvalue problem of the six-vertex model discussed in the present paper.
Nevertheless, it is still worth having a closer look at such conserved quantities living in the spectrum of the six-vertex
model, namely the functions $\nabla_{\vec{r}}^{(i,j)} (x)$. This is the goal of this section and in what follows we shall inspect the explicit
expressions for the first terms of \eqref{expan} in the $\mathfrak{h}$-modules with labels $n = 1, 2$.

\subsection{Case $n=1$} The resolution of the auxiliary linear problem for $n = 1$ was discussed in details in \Secref{sec:N1} and one can use 
the explicit formulae \eqref{v0v1} to evaluate our conserved quantities. For instance, in that case we find
$\nabla_0^{(0,1)} (x) = \log{\left( \frac{\mathcal{G}_{+} (x)}{\mathcal{G}_{-} (x)} \right)}$ with
\< \label{delta0}
\mathcal{G}_{+} (x) &\coloneqq& \phi_1 \lambda_{\mathcal{A}}(x) \left[ \cosh{(\gamma - x)} \lambda_{-}(0) + \sinh{(\gamma - x)} \lambda_{-}^{\prime} (0)  \right] \nonumber \\
&& \;  + \; \phi_2 \lambda_{\mathcal{D}}(x) \left[ \cosh{(\gamma + x)} \lambda_{-}(0) + \sinh{(\gamma + x)} \lambda_{-}^{\prime} (0)  \right] \nonumber \\
&& \; - \; \Lambda(x) \left[ \cosh{(x)} \lambda_{-}(0) + \sinh{(x)} \lambda_{-}^{\prime} (0)  \right] \nonumber \\
\mathcal{G}_{-} (x) &\coloneqq& \lambda_{-}(0) \left[ \phi_1 \sinh{(x - \gamma )} \lambda_{\mathcal{A}}(x) + \phi_2 \sinh{(x + \gamma )} \lambda_{\mathcal{D}}(x) - \sinh{(x)} \Lambda(x) \right] \; . \nonumber \\
\>
In \eqref{delta0} we have written $f^{\prime} (x) \coloneqq \frac{\dd f(x)}{\dd x}$ for the derivative of a given function $f$.
The explicit formula for the next term $\nabla_1^{(0,1)} (x)$ is already quite extensive, although not prohibitive, for generic parameters
$\phi_i$ and $\mu_i$. We shall not present it here in order to avoid an overcrowded section. 
As for the quantity $\nabla_0^{(0,1)}$ we find important to remark that $\partial \nabla_0^{(0,1)} (x) \propto \Sigma_1$ with $\Sigma_1$ the  differential function defined previously in \eqref{sig1}.

\begin{rema}
The substitution of solution \eqref{sol1} in \eqref{delta0} yields the relation
\[
\mathrm{e}^{\nabla_0^{(0,1)}} = \coth{(w_1)} + \frac{\lambda_{-}^{\prime} (0)}{\lambda_{-}(0)} \; ,
\]
which corroborates $\nabla_0^{(0,1)}$ is a constant on the solutions manifold.
\end{rema}

\begin{rema}
Alternatively, one could have used the conserved quantity $\nabla_0^{(0,1)}$ to determine the solution $\Lambda$ without going through
the procedure described in \Secref{sec:N1}.
\end{rema}

\subsection{Case $n=2$} 
In that case Theorem \ref{CQ} gives us three generating functions of conserved quantities, namely $\Theta_{0, 1}$, $\Theta_{0, 2}$
and $\Theta_{1, 2}$. Such generating functions are defined in terms of transport functions through formula \eqref{gen}. Then, using expansion \eqref{expan}, one can find explicit
expressions for  conserved quantities $\nabla_{\vec{r}}^{(i,j)}$. Explicit expressions for our conserved quantities are already
quite extensive though, even for $\vec{r} = (0 , 0 )$, and we shall not present them here. Despite the latter issues their direct inspection reveals
$\nabla_{(0,0)}^{(0,1)} = \nabla_{(0,0)}^{(0,2)}$ while $\nabla_{(0,0)}^{(1,2)}$ is indeed genuine.

Such conserved quantities could also have been used for determining the solution $\Lambda$ without going through the procedure
of \Secref{sec:N2}. Although we are not presenting  explicit expressions for the aforementioned conserved quantities here, 
it is still worth having a closer look at the equations $\partial \nabla_{(0,0)}^{(0,1)} (x) = 0$ and $\partial \nabla_{(0,0)}^{(1,2)} (x) = 0$ 
originating from Theorem \eqref{CQ}. For instance, equation $\partial \nabla_{(0,0)}^{(1,2)} (x) = 0$ is readily satisfied
while $\partial \nabla_{(0,0)}^{(0,1)} (x) = 0$ corresponds to a standard Riccati equation. More precisely, as for the
latter we find
\[ \label{riccati2}
\sinh{(\gamma)} \bar{\mathcal{K}}(x) \; \partial \Lambda(x) = \mathcal{K}_0 (x)  + \mathcal{K}_1 (x)  \Lambda(x) + \mathcal{K}_2 (x)  \Lambda(x)^2 \; .
\]
The explicit form of the coefficients $\bar{\mathcal{K}}$ and $\mathcal{K}_i$ are quite involved for generic parameters
$\phi_i$ and $\mu_i$. Due to that we restrict ourselves to presenting such coefficients only for the specialization $\phi_1 = \phi_2 = 1$ and $\mu_i = 0$. 
In that case they drastically simplify and read as follow:
\< \label{kb}
\bar{\mathcal{K}}(x) &\coloneqq& \lambda_{\mathcal{A}}(x) \sinh{(x)} \left[ \Lambda_0 \sinh{(\gamma - x)} + \lambda_{\mathcal{A}}(0) \sinh{(\gamma + x)}  \right]  \nonumber \\
&& \;  + \;  \lambda_{\mathcal{D}}(x) \sinh{(x + \gamma)} \left[ \Lambda_0 \sinh{(x)} + \lambda_{\mathcal{A}}(0) \sinh{(x + 2\gamma)}  \right] 
\>

\< \label{k0}
\mathcal{K}_0 (x) &\coloneqq&  \lambda_{\mathcal{A}}(x)^2 \left[ \lambda_{\mathcal{A}}(0) \sinh{(x )}^2 - \Lambda_0 \sinh{(x - \gamma)}^2  \right]  \nonumber \\
&& \;  + \; \lambda_{\mathcal{D}}(x)^2 \left[ \lambda_{\mathcal{A}}(0) \sinh{(x + 2\gamma )}^2 - \Lambda_0 \sinh{(x + \gamma)}^2  \right]  \nonumber \\
&& \;  - \; \lambda_{\mathcal{A}}(x) \lambda_{\mathcal{D}}(x) \left[ \lambda_{\mathcal{A}}(0) \left( \cosh{(4 \gamma)} - \cosh{(2 \gamma)} \cosh{(2x + 2\gamma)} \right) \right. \nonumber \\
&& \qquad  \qquad\quad \;\qquad \left. - \; \Lambda_0 \left( \cosh{(4 \gamma)} - \cosh{(2 \gamma)} \cosh{(2x)} \right) \right] \nonumber \\
&& \;  + \; \sinh{(\gamma)} \cosh{(\gamma)} \left[ \lambda_{\mathcal{A}}^{\prime} (x) \lambda_{\mathcal{D}}(x) - \lambda_{\mathcal{A}}(x) \lambda_{\mathcal{D}}^{\prime }(x) \right] \nonumber \\
&& \qquad  \times  \left[ \lambda_{\mathcal{A}}(0) \left( \cosh{(2 \gamma)} - \cosh{(2x + 2\gamma)} \right)  - \Lambda_0 \left( \cosh{(2 \gamma)} - \cosh{(2x)} \right) \right] 
\\ \nonumber \\
\label{k1}
\mathcal{K}_1 (x) &\coloneqq& \lambda_{\mathcal{A}}^{\prime} (x) \sinh{(\gamma)} \sinh{(x)} \left[ \lambda_{\mathcal{A}}(0) \sinh{(\gamma + x)} + \Lambda_0 \sinh{(\gamma - x)} \right]  \nonumber \\
&& \;  + \; \lambda_{\mathcal{D}}^{\prime} (x) \sinh{(\gamma)} \sinh{(x + \gamma)} \left[ \Lambda_0 \sinh{(x)} - \lambda_{\mathcal{A}}(0) \sinh{(x + 2\gamma )}  \right] \nonumber \\
&& \;  + \; \lambda_{\mathcal{A}}(x) \left[ \lambda_{\mathcal{A}}(0) \left( \cosh{(2 \gamma)} - \cosh{(\gamma)} \cosh{( 2 x + \gamma)} \right) \right.  \nonumber \\
&& \qquad  \;\; \;\qquad - \; \left. \Lambda_0 \left( \cosh{(2 \gamma)} - \cosh{(\gamma)} \cosh{( 2 x - \gamma)} \right) \right] \nonumber \\
&& \;  + \; \lambda_{\mathcal{D}}(x) \left[ \lambda_{\mathcal{A}}(0) \left( \cosh{(2 \gamma)} - \cosh{(\gamma)} \cosh{( 2 x + 3 \gamma)} \right) \right. \nonumber \\
&& \qquad \;\; \;\qquad - \;  \left. \Lambda_0 \left( \cosh{(2 \gamma)} - \cosh{(\gamma)} \cosh{( 2 x + \gamma)} \right) \right]  
\\ \nonumber  \\
\label{k2}
\mathcal{K}_2 (x) &\coloneqq&  \lambda_{\mathcal{A}}(0) \sinh{(x + \gamma)}^2 - \Lambda_0 \sinh{(x)}^2  \; . \nonumber \\
\>
In Eqs. \eqref{kb} through \eqref{k2} we  have employed the notation $\Lambda_0 \coloneqq \Lambda(0)$. Also, we anticipate that 
having \eqref{riccati2} as a standard Riccati equation will pave the way for associating a stationary Schr\"odinger equation 
with the six-vertex model eigenvalue problem. This will be discussed in \Secref{sec:SCHROD}.

%%%%%%%%%%%%%%%%%%%%%%%%%%%%%%%%%%%%%%%%%%
\section{Discretization of the spectrum} \label{sec:DISCR}
%%%%%%%%%%%%%%%%%%%%%%%%%%%%%%%%%%%%%%%%%%%

The potential well is a classical example of quantum mechanical system where boundary conditions lead to the discretization of the energy spectrum.
Boundary conditions play a similar role in our problem and in this section we aim to describe how they enter the previously discussed differential
description of the six-vertex model's spectrum. The requirement of such extra conditions can be anticipated since solutions of differential equations
are not unique in general. Unique solutions are then  characterized by initial or boundary value problems. As for the differential equations discussed in \Secref{sec:AUX}, such boundary conditions
will follow from the structure of the operators $\mathcal{A}$ and $\mathcal{D}$ forming the transfer matrix \eqref{tmat}.
The relevant structure of those operators has been already discussed in \cite{Galleas_2008} and \cite{Galleas_2015}. Here we only repeat that analysis
for the sake of completeness.

\begin{defi}[Change of variables] Let us introduce new variables $u$, $u_i$, $v_i$ and $q$ as
\begin{align}
u &\coloneqq e^{2 x} & u_i &\coloneqq e^{2 x_i} \nonumber \\
v_i & \coloneqq e^{2 \mu_i} & q &\coloneqq e^{\gamma} \; .
\end{align}
\end{defi}
Considering the six-vertex model $\mathcal{R}$-matrix \eqref{rmat} we then write
\[ \label{Rpol}
\mathcal{R}_{0 j} (x) = \bar{\mathcal{R}}_{0 j} (u) \eqqcolon \frac{1}{2} \begin{pmatrix}  u^{-\frac{1}{2}} \mathbb{a}_j (u) &  \mathbb{b}_j (u) \\ \mathbb{c}_j (u) &  u^{-\frac{1}{2}} \mathbb{d}_j (u) \end{pmatrix}
\]
in such a way that matrix representations $\mathbb{a}_j, \mathbb{b}_j, \mathbb{c}_j, \mathbb{d}_j \colon \C \to \mathrm{End}(\mathbb{V}^{\otimes L})$ are
explicitly given by
\begin{align} \label{pol_abcd}
\mathbb{a}_j (u) & = \begin{pmatrix} u q - q^{-1} & 0 \\ 0 & u - 1  \end{pmatrix}_j   & \mathbb{b}_j (u) & = \begin{pmatrix} 0 & 0 \\ q - q^{-1} & 0 \end{pmatrix}_j  \nonumber \\
\mathbb{c}_j (u) & = \begin{pmatrix} 0 & q - q^{-1} \\ 0 & 0 \end{pmatrix}_j   &  \mathbb{d}_j (u) & = \begin{pmatrix} u  - 1 & 0 \\ 0 & u q - q^{-1}  \end{pmatrix}_j \; . 
\end{align}
We have used tensor leg notation in \eqref{pol_abcd} and it is worth noticing that $\mathcal{R}_{0 j} (x - \mu_i) = \bar{\mathcal{R}}_{0 j} (u / w_i)$.

\smallskip
\noindent {\bf{Coproducts.}} Next we consider two copies of the Yang-Baxter algebra $\AR$ associated with the same six-vertex model 
$\mathcal{R}$-matrix \eqref{rmat}. Analogously to \eqref{ABCD}, we write
\< \label{ABCD12}
\mathcal{L}^{(1)} = \begin{pmatrix}  \mathcal{A}^{(1)} & \mathcal{B}^{(1)}  \\ \mathcal{C}^{(1)}  & \mathcal{D}^{(1)}  \end{pmatrix} \qquad \mbox{and} \qquad \mathcal{L}^{(2)} = \begin{pmatrix}  \mathcal{A}^{(2)}  & \mathcal{B}^{(2)}  \\ \mathcal{C}^{(2)}  & \mathcal{D}^{(2)}  \end{pmatrix}  
\>
for the generators of each copy. 
Coproducts $\Delta \colon \AR \to \AR \otimes \AR$ are then given by the standard formulae
\begin{align} \label{coprod}
\Delta (\mathcal{A}) &= \mathcal{A}^{(1)} \otimes   \mathcal{A}^{(2)}  + \mathcal{B}^{(1)} \otimes  \mathcal{C}^{(2)}  & \Delta (\mathcal{B}) &= \mathcal{A}^{(1)} \otimes   \mathcal{B}^{(2)}  + \mathcal{B}^{(1)} \otimes  \mathcal{D}^{(2)}  \nonumber \\
\Delta (\mathcal{C}) &= \mathcal{C}^{(1)} \otimes   \mathcal{A}^{(2)}  + \mathcal{D}^{(1)} \otimes  \mathcal{C}^{(2)}  & \Delta (\mathcal{D}) &= \mathcal{C}^{(1)} \otimes   \mathcal{B}^{(2)}  + \mathcal{D}^{(1)} \otimes  \mathcal{D}^{(2)}   \; .
\end{align}

\begin{lem} Let the pair $(\mathbb{V}^{\otimes L} , \widetilde{\mathcal{T}}_0 )$ be the $\AR$-module defined in \eqref{mono} associated 
with the six-vertex model $\mathcal{R}$-matrix \eqref{rmat}. Then, recalling the representation map 
$\pi \colon \AR \to \mathrm{End} (\mathbb{V}^{\otimes L})$, we have
\begin{align} \label{polAD}
\pi \left( \mathcal{A}(x) \right) & = u^{- \frac{L}{2}} \bar{\mathcal{A}}(u) & \pi \left( \mathcal{B}(x) \right) &= u^{- \frac{(L-1)}{2}} \bar{\mathcal{B}}(u)  \nonumber \\
\pi \left( \mathcal{C}(x) \right) &= u^{- \frac{(L-1)}{2}} \bar{\mathcal{C}}(u) & \pi \left( \mathcal{D}(x) \right) &= u^{- \frac{L}{2}} \bar{\mathcal{D}}(u) \; ,
\end{align}
where $\bar{\mathcal{A}}(u)$, $\bar{\mathcal{B}}(u)$, $\bar{\mathcal{C}}(u)$ and  $\bar{\mathcal{D}}(u)$ are polynomials in $u$.
In particular,  $\bar{\mathcal{A}}(u)$ and $\bar{\mathcal{D}}(u)$ are of degree $L$ while $\bar{\mathcal{B}}(u)$ and $\bar{\mathcal{C}}(u)$
are of degree $L-1$.
\end{lem}

\begin{proof}
Write $\V^{\otimes L} = \V^{(1)} \otimes \V^{(2)}$ with $\V^{(1)} = \V^{\otimes (L-1)}$ and $\V^{(2)} = \V$; and consider sub-modules
$(\mathbb{V}^{(1)} , \mathcal{L}^{(1)} )$ and $(\mathbb{V}^{(2)} , \mathcal{L}^{(2)} )$ with similar definition as \eqref{mono}.
Also, let us add the index $L$ to $\mathcal{A}$, $\mathcal{B}$, $\mathcal{C}$ and $\mathcal{D}$ in order to emphasize the associated
diagonalizable module $\V_{\mathcal{Q}} = \V^{\otimes L}$. More precisely, we write $\mathcal{A}_L$, $\mathcal{B}_L$, $\mathcal{C}_L$
and $\mathcal{D}_L$ to denote the matrix entries of \eqref{ABCD}.
Then using the coproducts \eqref{coprod} we find the following recursion relations,
\< \label{recur}
\mathcal{A}_L &=& \frac{1}{2} \left[ u^{-\frac{1}{2}} \left( \mathcal{A}_{L-1} \otimes 1 \right) \mathbb{a}_L (u) + \left( \mathcal{B}_{L-1} \otimes 1 \right) \mathbb{c}_L (u) \right] \nonumber \\
\mathcal{B}_L &=& \frac{1}{2} \left[ \left( \mathcal{A}_{L-1} \otimes 1 \right) \mathbb{b}_L (u) + u^{-\frac{1}{2}} \left( \mathcal{B}_{L-1} \otimes 1 \right) \mathbb{d}_L (u) \right] \nonumber \\
\mathcal{C}_L &=& \frac{1}{2} \left[ u^{-\frac{1}{2}} \left( \mathcal{C}_{L-1} \otimes 1 \right) \mathbb{a}_L (u) + \left( \mathcal{D}_{L-1} \otimes 1 \right) \mathbb{c}_L (u) \right] \nonumber \\
\mathcal{D}_L &=& \frac{1}{2} \left[ \left( \mathcal{C}_{L-1} \otimes 1 \right) \mathbb{b}_L (u) + u^{-\frac{1}{2}}  \left( \mathcal{D}_{L-1} \otimes 1 \right) \mathbb{d}_L (u) \right] \; .
\>
For $L=1$ we simply find that \eqref{polAD} holds by comparing \eqref{mono}, \eqref{ABCD}, \eqref{Rpol} and \eqref{pol_abcd}. The proof of \eqref{polAD}
for $L > 1$ then follows from \eqref{recur} by induction on $L$ using \eqref{pol_abcd}.
\end{proof}

\begin{cor}
From \eqref{tmat} and \eqref{polAD} we can conclude that $\pi \left( \mathrm{T}(x) \right) = u^{- \frac{L}{2}} \bar{\mathrm{T}}(u)$ with $\bar{\mathrm{T}}(u)$ 
a polynomial in $u$ of degree $L$.
\end{cor}

\begin{lem} \label{BVP}
The eigenvalues $\Lambda(x)$ of the transfer matrix \eqref{tmat} are of the form $\Lambda(x) = u^{- \frac{L}{2}} \bar{\Lambda}(u)$, where 
$\bar{\Lambda}(u)$ is a polynomial in $u$ of maximum degree $L$.
\end{lem}
\begin{proof}
Due to the Yang-Baxter equation \eqref{ybe} satisfied by the $\mathcal{R}$-matrix \eqref{rmat} we have the property 
$\left[ \mathrm{T}(x), \mathrm{T}(y)  \right] = 0$. Hence, there exists a basis of eigenvectors of $\mathrm{T}(x)$ independent of $x$.
The renormalized eigenvalue equation then reads $\bar{\mathrm{T}}(u) \ket{\Psi} = \bar{\Lambda}(u) \ket{\Psi}$ and the dependence of 
$\bar{\Lambda}(u)$ on $u$ is solely due to $\bar{\mathrm{T}}(u)$. Therefore, $\bar{\Lambda}$ is necessarily a polynomial and its maximum 
degree equals the degree of $\bar{\mathrm{T}}$. This concludes our proof and establishes boundary conditions associated with
the differential equations described in \Secref{sec:AUX}.
\end{proof}

Lemma \ref{BVP} completes the boundary value problem for the eigenvalues of the six-vertex model's transfer matrix \eqref{tmat}.
In particular, notice the function $\Lambda(x)$ given by \eqref{lambdaN} exhibits poles when $x \to w_i$. However, as polynomials, the 
eigenvalues $\Lambda$ can only exhibit poles at infinity. Therefore, the poles at $x = w_i$ must be a removable singularity and the residue
of \eqref{lambdaN} at these poles must vanish. The latter condition imposes the following constraint on the parameters $w_i$, namely
\[ \label{BAE}
\frac{\lambda_{\mathcal{A}}(w_i)}{\lambda_{\mathcal{D}}(w_i)} = (-1)^{n+1} \frac{\phi_2}{\phi_1} \prod_{\substack{j = 1 \\ j \neq i}}^{n} \frac{a(w_i - w_j)}{a(w_j - w_i)}  \; .
\]
As expected, the constraint \eqref{BAE} can be identified with standard \emph{Bethe ansatz equations}. 

\begin{rema}
There exist several solutions of \eqref{BAE} and each solution will then describe a different eigenvalue $\Lambda$.
Hence, through our approach we can clearly see the discretization of the transfer matrix's spectrum as a consequence of
the polynomial structure described in Lemma \ref{BVP}.
\end{rema}

\begin{rema}
The combined formulae \eqref{lambdaN} and \eqref{BAE} are very well known in the literature and they are usually obtained 
through \emph{Bethe ansatz} in all of its versions. However, within our approach the parameters $w_i$ are not introduced as an \emph{ansatz}. Here they are integration constants and follow
naturally from the resolution of differential equations. In particular, the number of such parameters is a consequence of the order of the associated
differential equation.
\end{rema}

%%%%%%%%%%%%%%%%%%%%%%%%%%%%%%%%%%%%%%%%%%
\section{From Riccati to Schr\"odinger} \label{sec:SCHROD}
%%%%%%%%%%%%%%%%%%%%%%%%%%%%%%%%%%%%%%%%%%%

Riccati equation, albeit non-linear, has a close connection with Sturm-Liouville problems. Moreover, when the coefficients of a Riccati
equation are properly fixed, one can find a rather simple map between Riccati and Schr\"odinger equations. This is what we intend to 
discuss in this section regarding  equations \eqref{riccati} and \eqref{riccati2}.

Eq. \eqref{riccati} describes the transfer matrix's eigenvalue $\Lambda$ in the $n=1$ $\mathfrak{h}$-module sector. In that case one can use the map
$\Lambda(x) = \alpha_1(x) \partial \log{(\psi(x))} + \beta_1(x)$ with
\<
\alpha_1 (x) &\coloneqq& \lambda_{-} (x) \sinh{(\gamma)} \nonumber \\
\beta_1 (x) &\coloneqq& \lambda_{+} (x) \cosh{(\gamma)} 
\>
to find the equation $\partial^2 \psi(x) - \psi(x) = 0$. The latter can be regarded as a Schr\"odinger equation with constant potential.

As for the $\mathfrak{h}$-module with label $n=2$ the situation is fortunately much more interesting. We proceed in an analogous manner and substitute
$\Lambda(x) = \alpha_2(x) \partial \log{(\psi(x))} + \beta_2(x) \bar{\beta_2}(x)$ in Eq. \eqref{riccati2}. The choices
\<
\alpha_2 &\coloneqq& \sinh{(\gamma)} \left( \Lambda_0 \sinh{(x)}^2 - \lambda_{\mathcal{A}}(0) \sinh{(x+\gamma)}^2   \right)^{-1} \nonumber \\
&& \times \left\{ \lambda_{\mathcal{A}}(x) \sinh{(x)} \left[ \Lambda_0 \sinh{(\gamma - x)} + \lambda_{\mathcal{A}}(0) \sinh{(\gamma + x)} \right]  \right. \nonumber \\ 
&& \quad + \;   \left. \lambda_{\mathcal{D}}(x) \sinh{(x+\gamma)} \left[ \Lambda_0 \sinh{(x)} - \lambda_{\mathcal{A}}(0) \sinh{(2\gamma + x)} \right]  \right\} \nonumber \\ 
\nonumber \\
\beta_2 &\coloneqq& \lambda_{\mathcal{A}}(x) \left\{ \sinh{(x)} \cosh{(\gamma)} \left[ \lambda_{\mathcal{A}}(0)^2 \sinh{(x+\gamma)}^3 - \Lambda_0^2 \sinh{(\gamma - x)} \sinh{(x)}^2 \right] \right. \nonumber \\
&& \qquad \qquad + \; 16^{-1} \lambda_{\mathcal{A}}(0) \Lambda_0 \left[ 1 - 2 \sinh{(2 \gamma)} \sinh{(4 x)} \right. \nonumber \\ 
&& \qquad \qquad \quad \;\; + \; 4 (\cosh{ (2 \gamma)} + 3 ) \cosh{(\gamma)}^2 \cosh{(2 x)} - 4 \cosh{(\gamma)}^2 \cosh{(4 x)} \nonumber \\
&& \qquad \qquad \quad \;\; + \; \left. \left. 8 \sinh{(\gamma)} \cosh{(\gamma)}^3 \sinh{(2 x)} - 14 \cosh{(2 \gamma)} + \cosh{(4\gamma)} \right] \right\} \nonumber \\
&& + \; \lambda_{\mathcal{D}}(x) \left\{ \sinh{(x + \gamma)} \cosh{(\gamma)} \left[ \lambda_{\mathcal{A}}(0)^2 \sinh{(x+ 2 \gamma)} \sinh{(x+\gamma)}^2 + \Lambda_0^2  \sinh{(x)}^3 \right] \right. \nonumber \\
&& \qquad \qquad + \; 16^{-1} \lambda_{\mathcal{A}}(0) \Lambda_0 \left[ 1 + 2 \cosh{(\gamma} \cosh{(\gamma - 2 x)} + 8 \cosh{(\gamma)} \cosh{(\gamma + 2 x)} \right. \nonumber \\
&& \qquad \quad \qquad \qquad \qquad \quad \quad + \;  6 \cosh{(\gamma)} \cosh{(3\gamma + 2 x)} - 4 \cosh{(\gamma)} \cosh{(3 \gamma + 4x)} \nonumber \\
&& \qquad \quad \qquad \qquad \qquad \qquad - \; \left. \left. 14 \cosh{(2 \gamma)} + \cosh{(4 \gamma)} \right] \right\} \nonumber \\
\nonumber \\
\bar{\beta_2}(x) &\coloneqq& \left( \Lambda_0 \sinh{(x)}^2 - \lambda_{\mathcal{A}}(0) \sinh{(x+\gamma)}^2 \right)^{-2}  \nonumber \\
\>
render the following Schr\"odinger equation for the function $\psi$,
\[ \label{schroedinger}
\partial^2 \psi(x) - \frac{3 c^2 }{\left[ \omega_0 b(x)^2 - \omega_0^{-1} a(x)^2 \right]^2} \psi(x) = \psi(x) \; .
\]
In \eqref{schroedinger} we have written $\Lambda_0 \eqqcolon c^L \omega_0^2$ and we can also identify
\[ \label{Vx}
\mathscr{V}(x) \coloneqq - \frac{3 c^2 }{\left[ \omega_0 b(x)^2 - \omega_0^{-1} a(x)^2 \right]^2}
\]
as the potential. In particular, it is also worth remarking that the stationary Schr\"odinger equation associated with the
potential \eqref{Vx} would normally read
\[ \label{normalS}
\left( \partial^2 + \mathscr{V}(x) \right) \psi(x) = \mathscr{E} \psi(x) \; .
\]
Hence, Eq. \eqref{schroedinger} can be regarded as a Schr\"odinger equation with fixed \emph{energy}. More precisely, the function $\Lambda$
representing eigenvalues of the six-vertex model's transfer matrix is mapped to solutions of \eqref{normalS} having energy $\mathscr{E} = 1$.

As extensively discussed in \Secref{sec:DISCR}, the characterization of the spectrum of the transfer matrix \eqref{tmat} through 
functional/differential equations is not complete until we declare \emph{initial} or \emph{boundary conditions} for the eigenvalues
$\Lambda$. Such conditions are then responsible for the discretization of the spectrum in consonance with \eqref{tmat} being a 
finite-dimensional matrix. Interestingly, within our approach we can identify the origins of such conditions. While the algebra (Yang-Baxter)
is solely responsible for the functional/differential equations, the particular representation under consideration supply the required boundary conditions.
In this way there is the possibility that the spectrum of the Schr\"odinger operator $\partial^2 + \mathscr{V}(x)$, with potential 
$\mathscr{V}(x)$ given by \eqref{Vx}, is not discrete for generic parameters $\omega_0$ since the latter
contains the initial conditions associated with the six-vertex model's spectral problem. In order to fix such parameter one can use properties of the 
representation \eqref{rmat}.

\begin{thm}[Initial condition] \label{om0}
The parameter $\omega_0^2$ is a $L$-th root-of-unity.
\end{thm}
\begin{proof}
The proof is essentially the same as the one previously described in \cite{Galleas_2008}. We first set $\phi_1 = \phi_2 = 1$ and 
$\mu_i = 0$. Next we notice $\mathcal{R}_{0 j} (0) = \sinh{(\gamma)} \mathcal{P}_{0 j}$ with $\mathcal{P}_{0 j}$ the permutation matrix
satisfying $\mathcal{P}_{0 i}^2 = \mathrm{Id}$ and $\mathcal{P}_{0 j} \mathcal{P}_{0 i} = \mathcal{P}_{0 i} \mathcal{P}_{i j}$.
The latter properties allow one to show $\mathrm{T}(0) = \sinh{(\gamma)}^L \mathscr{O}$ where
\[
\mathscr{O} \coloneqq \mathcal{P}_{1 L} \mathcal{P}_{1 L-1} \dots \mathcal{P}_{1 2} \; .
\]
Using induction and properties of the permutation matrix $\mathcal{P}_{i j}$ one can then prove
\[ \label{ON}
\mathscr{O}^n = \PROD{1}{j}{n} \left( \mathcal{P}_{j L} \mathcal{P}_{j L-1} \dots \mathcal{P}_{j n+1} \right) \; .
\]
Therefore, by setting $n=L$ in \eqref{ON}, we find $\mathscr{O}^L = \mathrm{Id}$ and consequently $\left( \omega_0^2 \right)^L = 1$. 
This concludes our proof. 
\end{proof}

As far as the potential function \eqref{Vx} is concerned, it is also interesting to examine the conditions on which
$\mathscr{V}(x) \in \R$. For instance, let us first assume $x, \gamma \in \R$. Next we would like to fix $\omega_0$ satisfying
Theorem \ref{om0} such that $\mathscr{V}(x)$ remains real. We can readily see $\omega_0 = \pm 1$ satisfy our requirements, and so does
$\omega_0 = \pm \ii$ for $L \in 2 \Z$. However, notice $\mathscr{V}(x)$ is invariant under the map $\omega_0 \to - \omega_0$ and
we can restrict our analysis to $\omega_0 = 1 , \ii$.  

Focusing first on the case $\omega_0 = \ii$, we find \eqref{Vx} consists of a potential barrier with shape depicted in Figure 2. 
The inspection of $\mathscr{V}(x)$ for a range of values of its parameters then suggests $\gamma$ controls both the shape of the barrier
as well as the location of its peak. In particular, we notice the barrier shape is very sensitive to $\gamma$ in the region
$0 < \gamma \lesssim 1$; while for $\gamma \gtrsim 1$ we see that $\gamma$ mainly shifts the center of the barrier to negative values
of $x$.

As for $\omega_0 = 1$ the behavior of \eqref{Vx} is drastically different. In that case $\mathscr{V}(x)$ corresponds to an infinite well
potential as shown in Figure 3. Similarly, in this case we also notice the shape of the potential well is much more sensitive to 
the parameter $\gamma$ in the region $0 < \gamma \lesssim 1$. Controlling the center of the well also seems to be the role of $\gamma$
in the region $\gamma \gtrsim 1$.

\begin{figure} \label{fig:OMI}
\centering
\subfloat[$\gamma = 0.1$]{\includegraphics[width=0.45\textwidth]{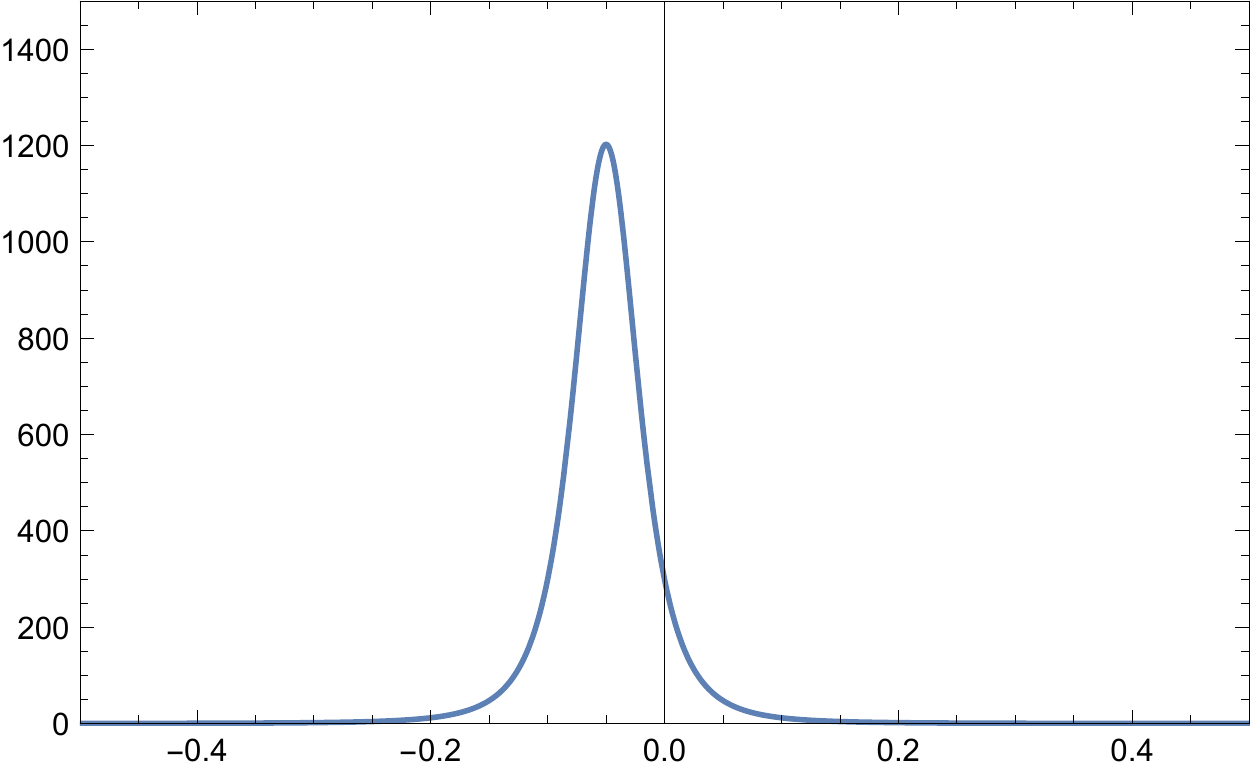}\label{fig:f1}}
\hfill
\subfloat[$\gamma = 0.3$]{\includegraphics[width=0.45\textwidth]{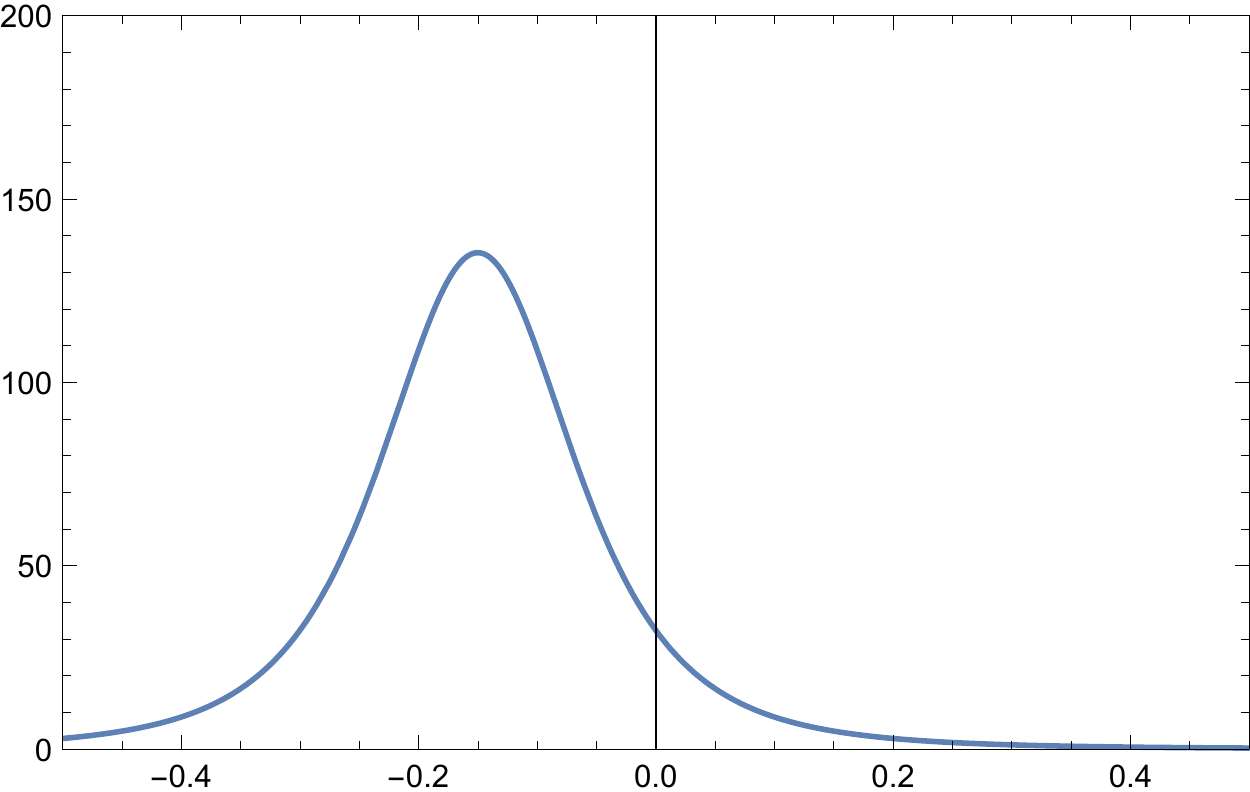}\label{fig:f2}}
\\
\subfloat[$\gamma = 5.43$]{\includegraphics[width=0.45\textwidth]{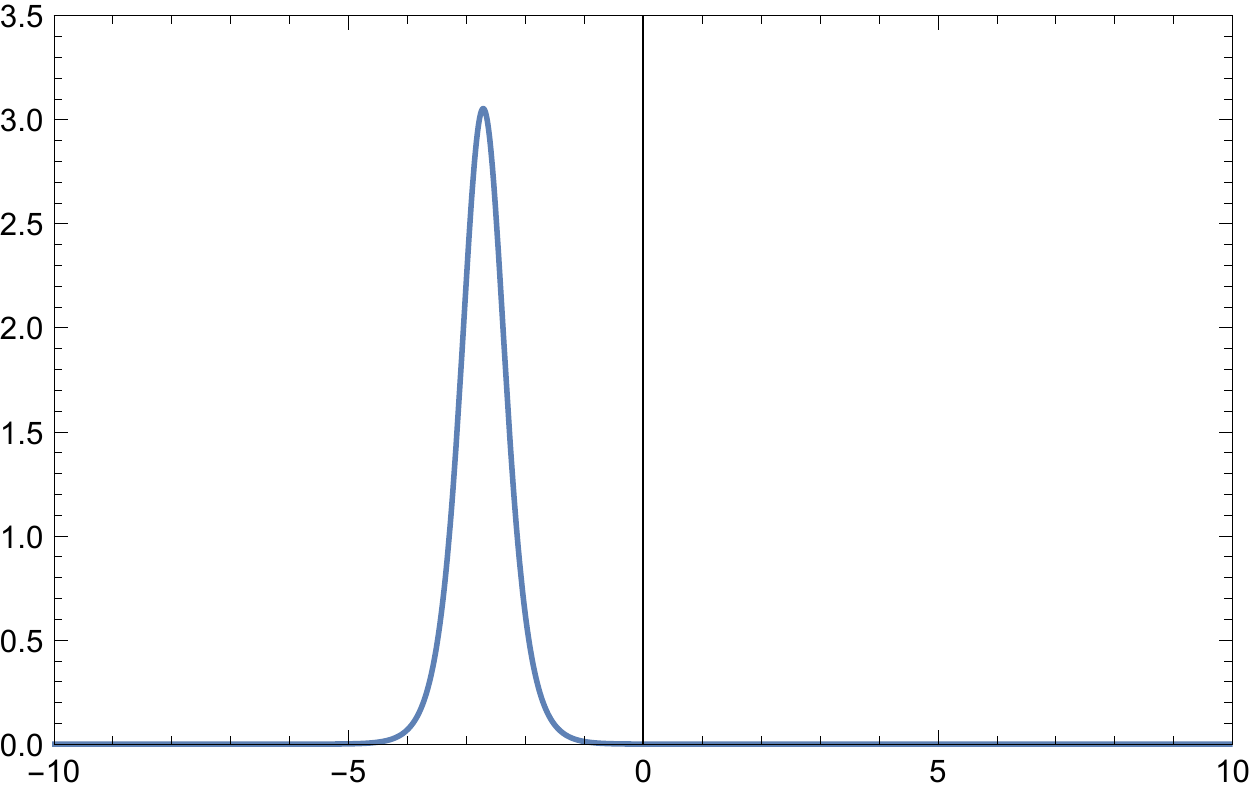}\label{fig:f3}}
\hfill
\subfloat[$\gamma = 8.12$]{\includegraphics[width=0.45\textwidth]{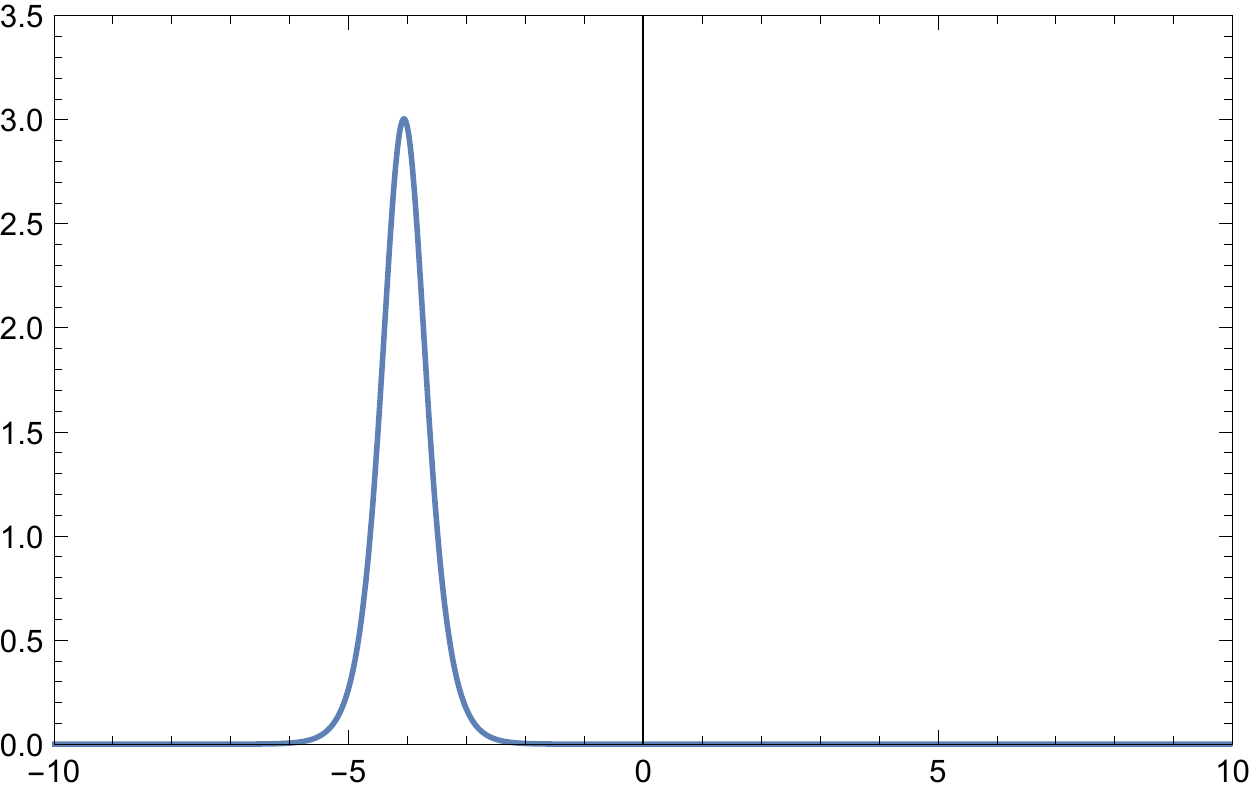}\label{fig:f4}}
\caption{Plot of the potential function \eqref{Vx} with $\omega_0 = \ii$ and given values of $\gamma$. The variable $x$ is represented in the horizontal axis while the vertical one denotes the potential $\mathscr{V}(x)$.}
\end{figure}

\begin{figure} \label{fig:OM1}
\centering
\subfloat[$\gamma = 0.1$]{\includegraphics[width=0.45\textwidth]{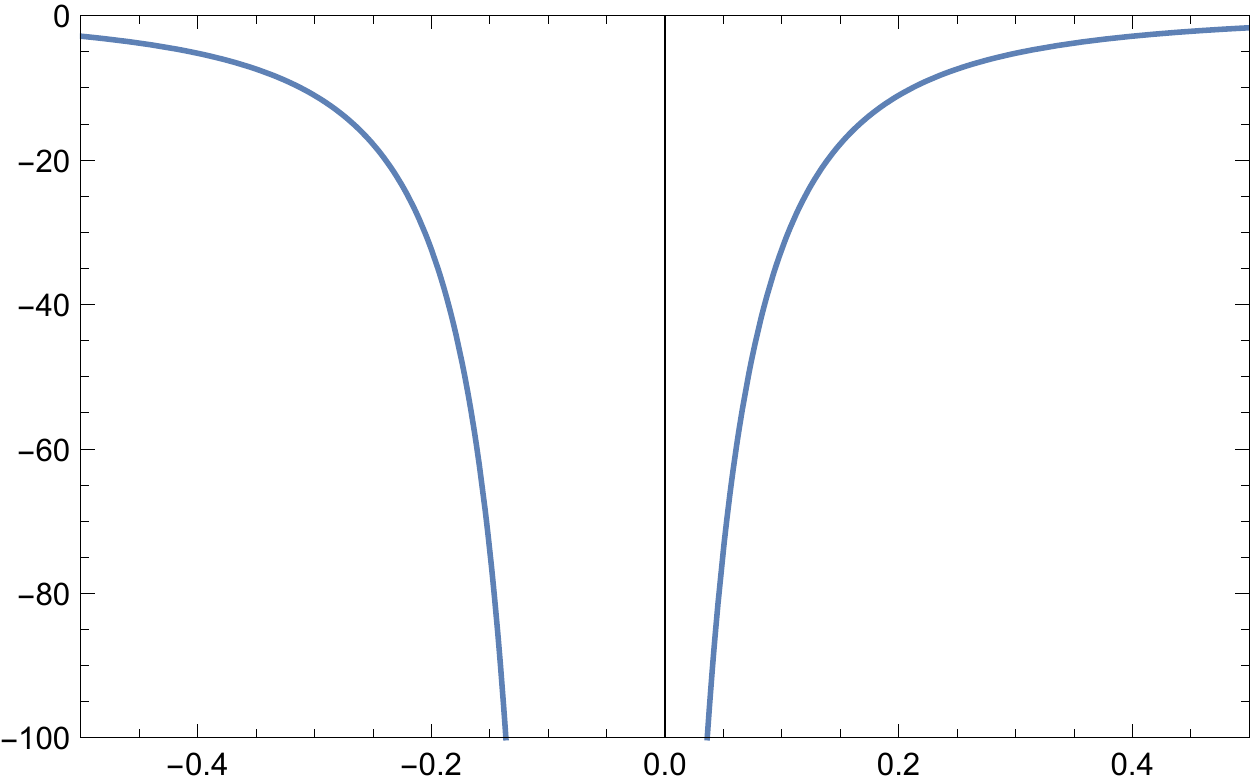}\label{fig:F1}}
\hfill
\subfloat[$\gamma = 0.3$]{\includegraphics[width=0.45\textwidth]{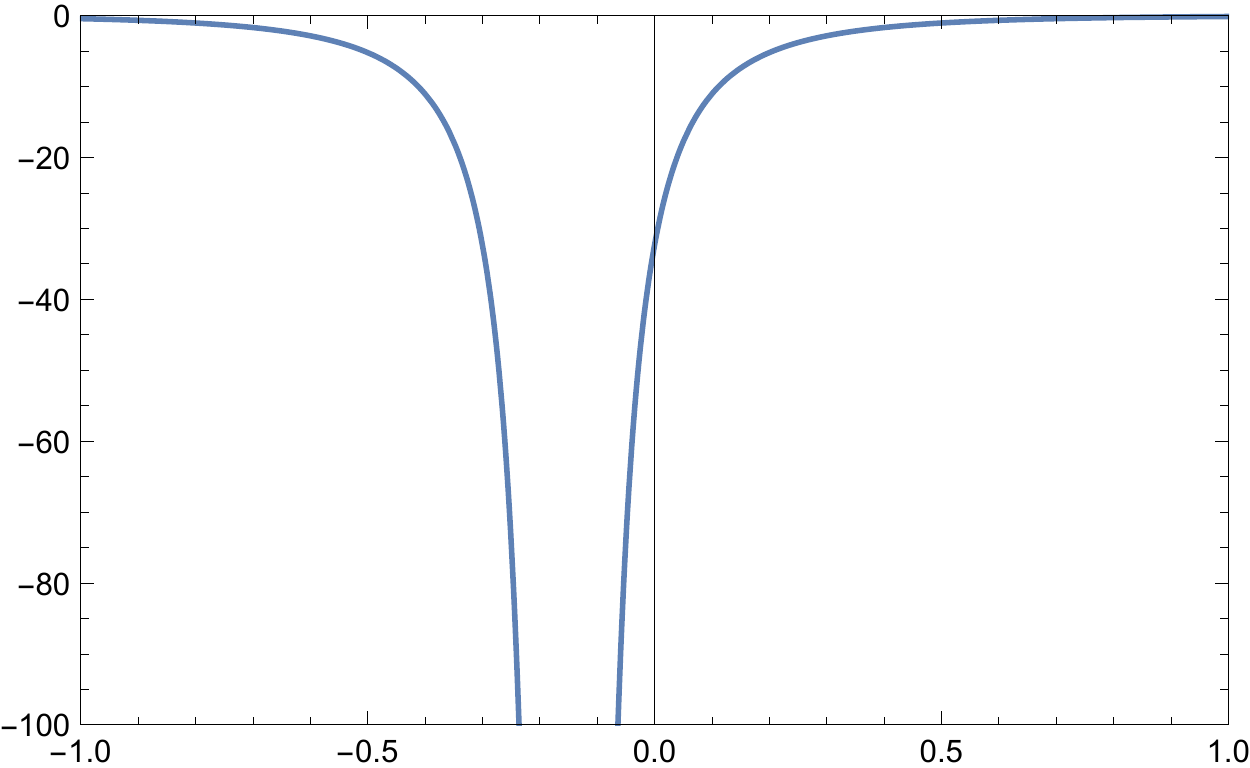}\label{fig:F2}}
\\
\subfloat[$\gamma = 5.43$]{\includegraphics[width=0.45\textwidth]{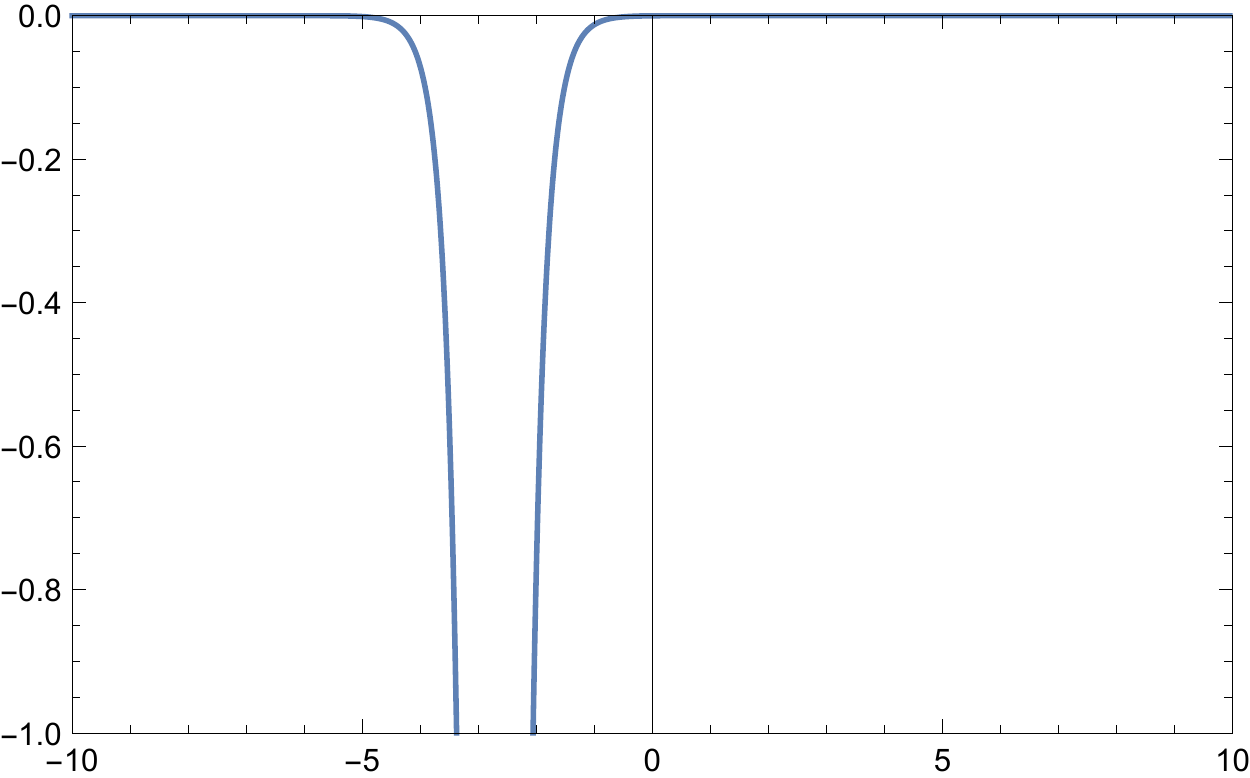}\label{fig:F3}}
\hfill
\subfloat[$\gamma = 8.12$]{\includegraphics[width=0.45\textwidth]{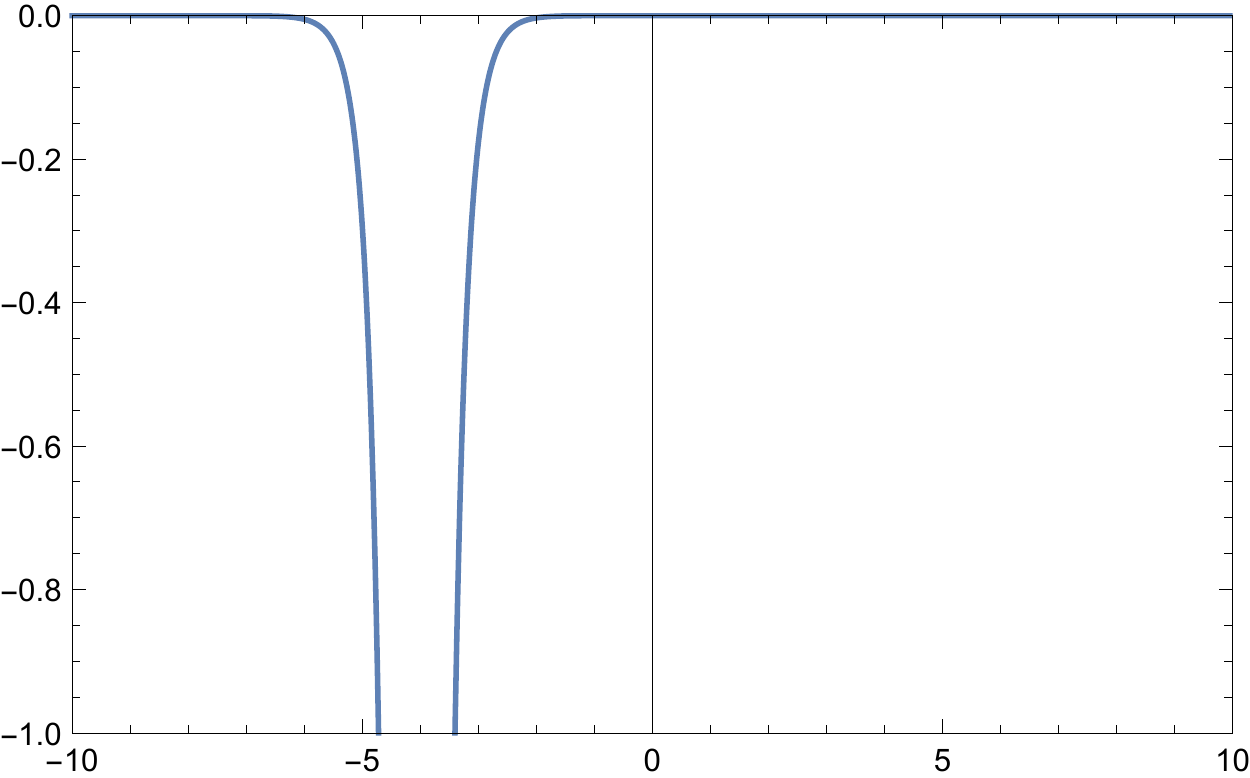}\label{fig:F4}}
\caption{Plot of the potential function \eqref{Vx} with $\omega_0 = 1$ and given values of $\gamma$. The variable $x$ is represented in the horizontal axis while the vertical one denotes the potential $\mathscr{V}(x)$.}
\end{figure}

%%%%%%%%%%%%%%%%%%%%%%%%%%%%%%%%%%%%%%%%%%
\section{Concluding remarks} \label{sec:CONCL}
%%%%%%%%%%%%%%%%%%%%%%%%%%%%%%%%%%%%%%%%%%%

In this paper we have reported on a series of relations between non-linear differential equations and the eigenvalue
problem of the six-vertex model's transfer matrix. The cornerstone of such relations is a formulation of the 
six-vertex model spectral problem as a \emph{boundary value problem} of certain non-linear differential equations. 

One important aspect of our analysis is the origin of such differential equations. The latter can be traced back to the 
Yang-Baxter algebra, which is a common algebraic structure underlying integrability of vertex models in the sense of Baxter.
In particular, we have shown that the mechanism leading to such formulation is in fact analogous to the framework of the 
\emph{classical inverse scattering method}. This analogy has motivated several results obtained in the present paper, for example,
the existence of generating functions of conserved quantities in the function space spanned by transfer matrix's eigenvalues.

The non-linear differential equations satisfied by eigenvalues of the six-vertex model's transfer matrix are \emph{Riccati equations} and they
can also be derived as specializations of certain non-linear functional equations describing the transfer matrix's spectrum. Riccati equation is one of
the first non-linear differential equations studied after the formulation of \emph{differential calculus} \cite{Riccati_1724}
and it has found several applications which includes the theory of conformal mapping through the Schwarzian 
derivative \cite{Ovsienko_Tabachnikov_book}. 
Interestingly, within our analysis such type of differential equation offers two routes for studying the transfer matrix's spectral problem. One of them is in 
terms of standard Bethe ansatz equations while the second option is through an underlying Schr\"odinger equation. We have elaborated on the latter
possibility in \Secref{sec:SCHROD}. Moreover, this Riccati representation might still open new possibilities for describing transfer matrix's
eigenvalues, without relying on Bethe ansatz equations, which we have not envisaged in the present paper. 

It is also important to remark that the description of certain eigenvalues associated with integrable vertex models through differential 
equations have appeared previously in the literature. For instance, Bazhanov and Mangazeev have shown in the works \cite{Bazhanov_Mangazeev_2005, Bazhanov_Mangazeev_2006, Bazhanov_Mangazeev_2010}
that a particular eigenvalue of the eight-vertex model \textsc{q}-operator satisfy certain differential equations when the model's
anisotropy parameter assumes a very special value. Here, however, we are describing eigenvalues of the transfer matrix and our analysis is
valid for generic values of the model's parameters.

Several results presented here do not depend on having a covering integrable lattice model. In particular, the auxiliary linear problem
put forward in Proposition \ref{aux} and the general mechanism leading to conserved quantities discussed in \Secref{sec:QUANT}. 
As a matter of fact, the auxiliary linear problem \eqref{MF} only relies on symmetric functions and the generating function \eqref{gen}
is a direct consequence of the formalism described in the present paper. In this way one could also think of embedding certain non-linear
differential equations within our approach along the lines of Lax representation of the KdV equation. That can be accomplished as follows. Firstly, one needs to find non-linear functional equations
such that the differential equation of interest is obtained as a specialization of the former. Next one needs to exhibit coefficients $\mathrm{M}_i$
such that Lemma \ref{det} produces the targeted non-linear functional equation.
However, finding such coefficients $\mathrm{M}_i$ might be a task far from trivial  and this is the point where the AF method plays an important role. As for 
integrable vertex models one can use the AF method to produce such coefficients encoding quantities of interest.

By using the AF method Eq. \eqref{MF} was also shown to accommodate several quantities including scalar products of Bethe vectors and partition functions
with special boundary conditions. For instance, by setting $\phi_1 = 1$, $\phi_2 = 0$,
$n=L$ and $\Lambda(x) = \lambda_{\mathcal{D}}(x)$; one can recognize \eqref{MF} with \eqref{MI} as the equation type A describing the partition
function of the six-vertex model with domain-wall boundaries \cite{Galleas_2016b}. A similar specialization also gives equation type D.	
Moreover, as discussed in \Secref{sec:scp} the symmetric functions $\mathcal{F}_n$ satisfying \eqref{MF} with \eqref{MI} also contains information
about the transfer matrix's eigenvectors. In a follow-up publication we intend to study in more details the role of the function
$\mathcal{F}_n$ in the six-vertex model.

%%%%%%%%%%%%%%%%%%%%%%%%%%%%%%%%%%%%%%%%%%%%%%%%%%%%%%%%
\bibliographystyle{alpha}
\bibliography{references}

\newpage
%%%%%%%%%%%%%%%%%%%%%%%%%%%%%%%%%%%%%%%%%%%%%%%%%%%%%%%%
\appendix

%%%%%%%%%%%%%%%%%%%%%%%%%%%%%%%%%%%%%%%%%%
\section{Functions $\xi_i (x)$} \label{app:FUN}
%%%%%%%%%%%%%%%%%%%%%%%%%%%%%%%%%%%%%%%%%%%

In this appendix we define the functions $\xi_i (x)$ entering the definition of the differential function $\Sigma_2$
in \Secref{sec:N2}. They read as follows:

\<
&&  \xi_0 (x) \coloneqq 2 (q^4 - 1)^2 \lambda_{+}^{\prime} (x) \lambda_{-}^{\prime} (x) \left[ (\lambda_{+} (x))^2  + (\lambda_{-} (x))^2  \right] \nonumber \\
&& \quad + \; (q^2 -1)^2 (q^4 -1) \left[ (\lambda_{+} (x))^2 \lambda_{-}^{\prime} (x) \lambda_{-}^{\prime \prime} (x) + (\lambda_{-} (x))^2 \lambda_{+}^{\prime} (x) \lambda_{+}^{\prime \prime} (x)   \right] \nonumber \\
&& \quad + \; 4 (q^2 -1) \left[ (q^2 +1)^2 (\lambda_{+}(x))^3 \lambda_{+}^{\prime} (x) \sum_{n=0}^2 q^n  + (q^2 - 1)^2 (\lambda_{-}(x))^3 \lambda_{-}^{\prime} (x) \sum_{n=0}^2 (-1)^n q^n \right] \nonumber \\
&& \quad + \; 2 (q^2 -1)^2 \sum_{n=0}^2 q^n \left[ (q^2 +1) (\lambda_{+}(x))^3 \lambda_{-}^{\prime \prime} (x) + (q - 1)^2 (\lambda_{-}(x))^3 \lambda_{+}^{\prime \prime} (x)  \right] \nonumber \\
&& \quad - \lambda_{+} (x) \lambda_{-} (x) \left\{   16 q^2 \sum_{n=0}^2 q^n \left[ (q - 1)^2 (\lambda_{+} (x))^2 - (q + 1)^2 (\lambda_{-} (x))^2   \right] \right. \nonumber \\
&& \quad \qquad + \; 2 (q^2 -1)^2 \sum_{n=0}^2 q^n \left[ (q - 1)^2 \lambda_{-} (x) \lambda_{-}^{\prime \prime} (x)   + (q + 1)^2 \lambda_{+} (x) \lambda_{+}^{\prime \prime} (x)   \right] \nonumber \\
&& \quad \qquad + \;  2 (q^4 - 1)^2 \left[ (\lambda_{+}^{\prime} (x))^2  + (\lambda_{-}^{\prime} (x))^2 \right] + (q^2 + 1) (q^2 - 1)^3 \left[ \lambda_{+}^{\prime} (x) \lambda_{-}^{\prime \prime} (x) + \lambda_{-}^{\prime} (x) \lambda_{+}^{\prime \prime} (x)  \right] \nonumber \\
&& \quad \qquad + \;  4 (q + 1)^2 (q^4 - 1) (1 - 3 q + q^2)  \lambda_{-} (x) \lambda_{+}^{\prime} (x) \nonumber \\
&& \quad \qquad \left. + \; 4 (q^2 -1) \sum_{n=0}^2 (-1)^n q^n \; (1 + 2 q + 6 q^2 + 2 q^3 + q^4) \lambda_{+} (x) \lambda_{-}^{\prime} (x)  \right\} \\
\nonumber \\
\nonumber \\
\nonumber \\
&& \xi_1 (x) \coloneqq  8 (1+q)^2 (\lambda_{-}(x))^3 \sum_{n=0}^2 q^{2 n} + (1 - q) \lambda_{-}(x) \left[ 2 \lambda_{+}(x) \sum_{n=0}^2 q^n - (1 + q^2) \lambda_{-}^{\prime} \right] \nonumber \\
&& \quad   + \;  2(q+1)^2 (q^2-1) \left[ 4(q^2+1) (\lambda_{-}(x))^2 \lambda_{+}^{\prime}  + (q^2-1) \lambda_{-}(x) ( \lambda_{+}^{\prime}(x) )^2  \right]  \nonumber \\
&& \quad  + \; (q^2 - 1) \lambda_{+}(x) \left[ 2(q^2 + 1) \lambda_{+}^{\prime} (x) + (q^2 - 1) \lambda_{-}^{\prime \prime} (x) \right] \left[ 2 \lambda_{+}(x) \sum_{n=0}^2 q^n + (q^2 -1) \lambda_{-}^{\prime} \right] \nonumber \\
&& \quad\quad \quad \quad \times \; \left[ (1 + q) (q^2 - 1) \lambda_{+}^{\prime \prime} (x) + 4 \left( (q -1) \lambda_{+}(x) + (q + 1) \lambda_{-}^{\prime} (x) \right) \sum_{n=0}^2 (-1)^n q^n \right] 
\>

\<
&& \xi_2 (x) \coloneqq  2 (q^2 -1) (3 + 2 q + 3 q^2) \lambda_{+} (x) \lambda_{-} (x) \lambda_{+}^{\prime} (x) \nonumber \\
&& \;\; \quad \quad \quad + \; (q^2 -1)^2 \lambda_{-}(x) \left[ \lambda_{-} (x) \lambda_{+}^{\prime \prime} (x) - \lambda_{+} (x) \lambda_{-}^{\prime \prime} (x)  \right]  \nonumber \\
&& \;\; \quad \quad  \quad + \; 4(1 + q^2) \lambda_{+} (x) \left[ (1 + 5 q + q^2) (\lambda_{-}(x))^2 - (\lambda_{+}(x))^2 \sum_{n=0}^2 q^n \right] \nonumber \\
&& \;\; \quad \quad  \quad \quad - \; 2 (q^2 - 1) \lambda_{-}^{\prime} \left[ (1 + q^2) (\lambda_{+}(x))^2 +  2 (\lambda_{-}(x))^2 \sum_{n=0}^2 (-1)^n q^n  \right] \\
\nonumber \\
&& \xi_3 (x) \coloneqq (q^2 -1) \left[ \lambda_{-} (x) \lambda_{+}^{\prime} (x) - \lambda_{+} (x) \lambda_{-}^{\prime} (x)  \right]  \nonumber \\
&& \qquad \qquad + \; 2 ( \lambda_{-} (x) )^2 \sum_{n=0}^2 (-1)^n q^n - 2 ( \lambda_{+} (x) )^2 \sum_{n=0}^2 q^n 
\>

\<
&&\xi_4 (x) \coloneqq   8 \lambda_{+}(x) \lambda_{-}(x) \sum_{n=0}^4 (-1)^n q^n  + 2 (q^2 -1) (3 q^2 - 2 q + 3) \lambda_{-}(x) \lambda_{-}^{\prime} (x)  \nonumber \\
&& \qquad \qquad  - \; 2 (q^4 -1) \lambda_{+}(x) \lambda_{+}^{\prime} (x)  + (q^2 - 1)^2 \left[ \lambda_{-}(x) \lambda_{+}^{\prime \prime} (x)  - \lambda_{+}(x) \lambda_{-}^{\prime \prime} (x)   \right] \\
\nonumber \\
&& \xi_5 (x) \coloneqq  2(q^2 +1) (\lambda_{-}(x))^2 - 2(q^2 -1) (\lambda_{+}(x))^2 \nonumber \\
&& \qquad \qquad + \; (q^2 -1) \left[ \lambda_{-}(x) \lambda_{+}^{\prime}(x) - \lambda_{+}(x) \lambda_{-}^{\prime}(x)    \right]  \; .
\>

\end{document}